\pdfoutput=1
\RequirePackage{ifpdf}
\ifpdf % We~are running pdfTeX in pdf mode
\documentclass[pdftex]{sigma}
\else
\documentclass{sigma}
\fi

\numberwithin{equation}{section}

\newtheorem{Theorem}{Theorem}[section]
\newtheorem*{Theorem*}{Theorem}
\newtheorem{Corollary}[Theorem]{Corollary}
\newtheorem{Lemma}[Theorem]{Lemma}
\newtheorem{Proposition}[Theorem]{Proposition}
 { \theoremstyle{definition}
\newtheorem{Definition}[Theorem]{Definition}

\newtheorem{Example}[Theorem]{Example}
\newtheorem{Remark}[Theorem]{Remark} }

\newcommand{\norm}[1]{\lVert#1\rVert}
\newcommand{\abs}[1]{\lvert#1\rvert}

\begin{document}
\allowdisplaybreaks

\renewcommand{\thefootnote}{}

\newcommand{\arXivNumber}{2304.13272}

\renewcommand{\PaperNumber}{007}

\FirstPageHeading

\ShortArticleName{A General Dixmier Trace Formula for the Density of States on Open Manifolds}

\ArticleName{A General Dixmier Trace Formula\\ for the Density of States on Open Manifolds\footnote{This paper is a~contribution to the Special Issue on Global Analysis on Manifolds in honor of Christian B\"ar for his 60th birthday. The~full collection is available at \href{https://www.emis.de/journals/SIGMA/Baer.html}{https://www.emis.de/journals/SIGMA/Baer.html}}}

\Author{Eva-Maria HEKKELMAN~$^{\rm a}$ and Edward MCDONALD~$^{\rm b}$}

\AuthorNameForHeading{E.~Hekkelman and E.~McDonald}

\Address{$^{\rm a)}$~School of Mathematics and Statistics, University of New South Wales,\\
\hphantom{$^{\rm a)}$}~Kensington, NSW 2052, Australia}
\EmailD{\href{mailto:e.hekkelman@unsw.edu.au}{e.hekkelman@unsw.edu.au}}

\Address{$^{\rm b)}$~Department of Mathematics, Penn State University, University Park, PA 16802, USA}
\EmailD{\href{mailto:eam6282@psu.edu}{eam6282@psu.edu}}

\ArticleDates{Received April 27, 2023, in final form January 10, 2024; Published online January 17, 2024}

\Abstract{We give an abstract formulation of the Dixmier trace formula for the density of states. This recovers prior versions and allows us to provide a Dixmier trace formula for the density of states of second order elliptic differential operators on manifolds of bounded geometry satisfying a certain geometric condition. This formula gives a new perspective on Roe's index on open manifolds.}


\Classification{47N50; 58J42; 53C23}

\renewcommand{\thefootnote}{\arabic{footnote}}
\setcounter{footnote}{0}

\section{Introduction}
%\label{S: Introduction}
Originating from solid-state physics, the density of states (DOS) is used to study electrical and thermal properties of a material. Loosely speaking, it should describe how many quantum states are admitted at each energy level per volume of the material studied. Mathematically, there are multiple ways to define it rigorously. In this paper, we define the DOS as follows, similar to~\mbox{\cite[Section C]{Simon1982}}. Given a (possibly unbounded) lower-bounded self-adjoint operator $H$ on the Hilbert space $L_2(X)$, where $X$ is some metric space with a Borel measure, we assume the existence of the limits
\[
\lim_{R\to \infty} \frac{1}{|B(x_0,R)|}\mathrm{Tr}\bigl(\exp(-tH)M_{\chi_{B(x_0,R)}}\bigr), \qquad t>0,
\]
where $|B(x_0,R)|$ denotes the volume (measure) of the closed ball $B(x_0,R)$ with center $x_0 \in X$ and radius $R$, $\mathrm{Tr}$ is the usual operator trace, $\exp(-tH)$ is obtained via functional calculus, $M_g$ denotes the multiplication operator with respect to the function $g$ and $\chi_{I}$ is the indicator function of the set $I$. The existence of these limits implies, via the Riesz--Markov--Kakutani theorem, the existence of a Borel measure $\nu_H$~\cite[Proposition C.7.2]{Simon1982} such that
\[
\lim_{R\to \infty} \frac{1}{|B(x_0,R)|}\mathrm{Tr}\big(f(H)M_{\chi_{B(x_0,R)}}\big) = \int_{\mathbb{R}}f \,{\rm d}\nu_H, \qquad f\in C_c(\mathbb{R}).
\]
This measure, if it exists, is what we call the density of states of the operator $H.$ This essentially coincides with similar definitions elsewhere in the literature, at least for Schr\"odinger operators on Euclidean space, see~\cite{Simon1982}.

Apart from its origin in physics, the density of states has been proposed by Strichartz~\cite{Strichartz2012} as a substitute for the spectral counting function $n(\lambda,H) := \mathrm{Tr}\big(\chi_{(-\infty,\lambda)}(H)\big)$ in the case where $H$ does not have discrete spectrum.

The DOS has been studied extensively in mathematical contexts. Common research areas are the existence of the DOS~\cite{AizenmanWarzel2015, BerezinShubin1991, CarmonaLacroix1990, DoiIwatsuka2001, PasturFigotin1992, Shubin1979, Simon1982}, the analytical properties of the function $\lambda \mapsto \nu_H((-\infty, \lambda])$~\cite{AizenmanWarzel2015, BourgainKlein2013, CarmonaLacroix1990, PasturFigotin1992} and its asymptotic behaviour as $\lambda$ approaches boundaries of the support of $\nu_H$~\cite{AizenmanWarzel2015, BourgainKlein2013, CarmonaLacroix1990, Lang1991, PasturFigotin1992}. The study of the Anderson localisation phenomenon is inseparably related to the study of the density of states~\cite{AizenmanWarzel2015, CarmonaLacroix1990, Lang1991}.

There are three obvious choices for the metric space $X$ on which we work. First of all, one can simply take $X = \mathbb{R}^d$ as a flat background space as is usually done. Secondly, one can consider a discrete metric space $X$ as a discrete model or a discrete approximation of reality~\cite{AizenmanWarzel2015, AlexanderOrbach1982, BourgainKlein2013, CarmonaLacroix1990, ChayesChayes1986, Hof1993, KirschMuller2006, PasturFigotin1992, Veselic2005, Wegner1981}. Finally, one can take $X$ to be a manifold to study the DOS in even greater generality~\cite{AdachiSunada1993, LenzPeyerimhoff2008, LenzPeyerimhoff2004, PeyerimhoffVeselic2002, Veselic2008}. The present paper will mainly focus on the manifold setting.

The previous papers~\cite{AHMSZ,AMSZ} prove the formula
\begin{equation}
\label{eq:general_form}
\mathrm{Tr}_\omega(f(H) M_w) = \int_{\mathbb{R}} f \, {\rm d}\nu_H, \qquad f\in C_c(\mathbb{R}),
\end{equation}
where $w$ is a fixed weight function, for a certain class of operators $H$ admitting a DOS on, respectively, Euclidean space and certain discrete metric spaces. Here $\mathrm{Tr}_\omega$ denotes a Dixmier trace on the ideal of weak trace-class operators $\mathcal{L}_{1,\infty} (L_2(X))$ (see Section~\ref{S: Prelims}). In this paper, we will extend formula~\eqref{eq:general_form} to a statement in abstract operator theory (Theorem~\ref{main_theorem} below). This abstract statement allows for a new proof of~\eqref{eq:general_form} avoiding the heavy real analytic difficulties encountered in~\cite{AHMSZ, AMSZ}. Furthermore, the abstract operator theoretic reformulation of the density of states formula allows us to generalise~\eqref{eq:general_form}
to certain manifolds of bounded geometry, for the definition of bounded geometry, we refer to Section~\ref{S: Manifolds}.

Formulas like~\eqref{eq:general_form} arise in a multitude of situations. First of all, there is a clear resemblance (and in fact, relation~\cite{AMSZ}) to Connes' trace formula~\cite{Connes1988, LMSZVol2} on Euclidean space
\[
\mathrm{Tr}_\omega\big(M_f (1-\Delta)^{-\frac{d}{2}}\big) = C_d \int_{\mathbb{R}^d} f(t) \, {\rm d}t, \qquad f \in C_c\big(\mathbb{R}^d\big),
\]
for some constant $C_d$ only depending on $d$. Connes' trace formula shows that integration with respect to the Lebesgue measure can be recovered with Dixmier traces. Formula~\eqref{eq:general_form} in turn shows that the same holds true for the DOS measure. Furthermore, a link between the DOS and Dixmier traces has previously appeared in the seminal work by Bellissard, van Elst and Schulz-Baldes~\cite{BellissardVanElst1994}. Noting the relation between Dixmier traces and $\zeta$-function residues~\cite{LSZVol1}, a result by Bourne and Prodan~\cite[Lemma~6.1]{BourneProdan2018} in particular bears some resemblance to the results in this article.

Let us go into more detail concerning the results~\cite{AHMSZ,AMSZ}. The paper~\cite{AMSZ} proves equation~\eqref{eq:general_form} on Euclidean space for Schr\"odinger operators of the form $H = -\Delta + M_V$, with $V\in L_\infty\big(\mathbb{R}^d\big)$ real-valued. The weight function $w$ is of the form \smash{$w(x) = \langle x \rangle^{-d} = \big(1+|x|^2\big)^{- d/2}$}.

In~\cite{AHMSZ}, the setting is discrete metric spaces $(X, d_X)$ satisfying a property called Property~(C)~by the authors. Namely, if $\{r_k\}_{k\in \mathbb{N}}$ denotes the sequence obtained by ordering the set $\{d_X\!(x_0, y)\}_{y{\in} X}$ in increasing manner, then the theorem in~\cite{AHMSZ} requires
\begin{equation}\tag{C}\label{eq:Property_C}
 \lim_{k\to \infty} \frac{|B(x_0, r_{k+1})|}{|B(x_0, r_{k})|} =1.
\end{equation}
With this assumption, equation~\eqref{eq:general_form} is then proven for general lower bounded self-adjoint operators $H$. The function $w$ takes the form
\[
w(x) := \frac{1}{1+|B(x_0, d_X(x_0, x))|}, \qquad x \in X.
\]
We remark that a weaker version of this result can be proven with the techniques developed in this manuscript, but we will not explicitly show this.

In the present paper, we will prove equation~\eqref{eq:general_form} for uniformly elliptic differential operators on certain manifolds of bounded geometry. Like in the discrete case, an additional geometrical condition is needed for the main theorem. We would now first like to discuss this condition.

Let $X$ be a $d$-dimensional Riemannian manifold (assumed to be smooth, orientable and connected throughout this paper), and let $d_X$ be the distance function on $X$ induced by the Riemannian metric. The Riemannian volume of the ball $B(x_0, r)$ is denoted by $|B(x_0, r)|$. Its boundary, $\partial B(x_0, r)$, is a $(d-1)$-dimensional Hausdorff-measurable subset of $X$, and as such we can talk about its volume, calculated with respect to its inherited Riemannian metric. This $(d-1)$-dimensional volume we will also denote as $|\partial B(x_0, r)|$. In fact, it then holds that (see Section~\ref{S: Manifolds})
\[
\frac{{\rm d}}{{\rm d}r}|B(x_0, r)|=|\partial B(x_0, r)|.
\]

What we require of our manifolds is that both $|B(x_0,r)|$ and $|\partial B(x_0,r)|$ grow sufficiently slowly and regularly. Namely, we will ask that the ratios $\frac{|\partial B(x_0,R)|}{|B(x_0,R)|}$ and $\frac{\frac{{\rm d}}{{\rm d}r}|_{r=R}|\partial B(x_0,r)|}{|\partial B(x_0,R)|}$ vanish as~${R\to\infty}$ in the following way.

\begin{Definition}[Property~(D)]\label{D:Property_D}
 Let $X$ be a non-compact Riemannian manifold of bounded geometry. It is said to have Property~(D) if
\begin{equation*}%\label{Condition1b}
\left \{ \frac{\abs{\partial B(x_0, k)}}{\abs{B(x_0,k)}}\right\}_{k\in \mathbb{N}} \in \ell_2(\mathbb{N})
\end{equation*}
and
\begin{equation*}%\label{Condition2}
 \lim_{R\to \infty} \frac{\frac{\partial}{\partial r}\big |_{r=R} |\partial B(x_0, r)|}{|\partial B(x_0,R)|} = 0.
\end{equation*}
\end{Definition}

Note that if a function $f\in C^1(\mathbb{R})$ satisfies $\lim_{x\to \infty}\frac{f'(x)}{f(x)}= 0$, then $\log f(x) = o(x)$ and hence $f(x) = {\rm e}^{o(x)}$. Therefore, if the manifold $X$ satisfies Property~(D),
 then necessarily both $|B(x_0,r)| ={\rm e}^{o(r)}$ and $|\partial B(x_0, r)| = {\rm e}^{o(r)}$. A quick calculation shows that Property~(D) still admits volume growth of the order $|B(x_0, r)| = \exp \big(r^{\frac{1}{2}-\varepsilon}\big)$. The conditions listed mostly serve to prevent erratic behaviour of the growth. Observe the similarity in this sense to Property~(C) for discrete metric spaces quoted above. Property~(D) is obviously satisfied for Euclidean spaces.

To add, in the cited literature on the DOS in a manifold setting, one frequently has a manifold~$X$ with a discrete, finitely generated group $\Gamma$ of isometries of $(X,g)$ which acts freely and properly discontinuously on $X$ such that the quotient $X/\Gamma$ is compact~\cite{AdachiSunada1993, LenzPeyerimhoff2008, LenzPeyerimhoff2004, PeyerimhoffVeselic2002, Veselic2008}. We shall treat this example in more detail in Section~\ref{S: Example}. In each of the cited papers, a recurring assumption is that $\Gamma$ is amenable. This turns out to be equivalent with the existence of an expanding family of bounded domains $D_j\subset X$ such that
\[
\lim_{j \to \infty} \frac{|\partial_h D_j|}{|D_j|} = 0, \qquad \forall h>0,
\]
where $\partial_h D_j := \{x \in D_j \colon d(x, \partial D_j) \leq h \}$~\cite{AdachiSunada1993}. This is clearly a closely related property, and we will show that it is weaker than Property~(D), i.e., Property~(D) in this setting implies that $\Gamma$ is amenable.

To formulate the main theorem, we need to specify the class of operators for which it works. The following definition is essentially the same as~\cite{Kordyukov1991} and the $C^\infty$-bounded differential operators defined in~\cite[Appendix 1]{Shubin1992}.
\begin{Definition}%\label{D:BD}
A differential operator $P$ on a $d$-dimensional Riemannian manifold $X$ of bounded geometry is called a uniform differential operator (of order $m$) if it can be expressed in normal coordinates in a neighbourhood of each point $x \in X$ as
\[
P = \sum_{|\alpha|\leq m} a_{\alpha, x}(y) D^\alpha,\qquad D^{\alpha} = {\rm i}^{-|\alpha|}\partial_{y_1}^{\alpha_1}\cdots \partial_{y_d}^{\alpha_d}
\]
and for all multi-indices $\beta$, we have
\[
\big\lvert D^\beta a_{\alpha, x}(0)\big\rvert \leq C_{\alpha, \beta}, \qquad \abs{\alpha}\leq m.
\]
Following the notation of~\cite{Kordyukov1991}, we denote this by $P \in BD^m(X)$.

Furthermore, let $\sigma_x(y, \xi) = \sum_{\abs{\alpha}=m} a_{\alpha, x}(y) \xi^\alpha$ be the principal symbol of $P$ near $x.$ We say that $P \in BD^m(X)$ is uniformly elliptic, denoted $P\in EBD^m(X)$, if there exists $\varepsilon>0$ such that
\[
\abs{\sigma_x(0,\xi)}\geq \varepsilon \abs{\xi}^m,\qquad \xi\in \mathbb{R}^d,\quad x\in X.
\]
\end{Definition}
A similar definition applies to operators acting between sections of vector bundles of bounded geometry; see~\cite{Shubin1992}.

\begin{Theorem}\label{T: main manifold thm}
 Let $(X,g)$ be a non-compact Riemannian manifold of bounded geometry with Property~$($D$)$. Let $P \in EBD^2(X)$ be self-adjoint and lower-bounded, and let $w$ be the function on $X$ defined by \[ w(x) = (1+\abs{B(x_0,d_X(x,x_0))})^{-1},\qquad x \in X.\]

 Then $f(P)M_w$ is an element of $\mathcal{L}_{1,\infty}$ for all compactly supported functions $f\in C_c(\mathbb{R})$. If $P$ admits a density of states $\nu_P$, we have for all extended limits $\omega$
 \[
 \mathrm{Tr}_{\omega}(f(P)M_w) = \int_{\mathbb{R}} f(\lambda)\, {\rm d}\nu_P(\lambda), \qquad f\in C_c(\mathbb{R}).
 \]
\end{Theorem}
For the definition of the space of weak trace-class operators $\mathcal{L}_{1,\infty}$, see Section~\ref{S: Prelims}.

Note that manifolds $X$ with a group $\Gamma$ acting freely and properly discontinuously on $X$ by isometries such that $X/\Gamma$ is compact are of bounded geometry (see Section~\ref{S: Manifolds}), although for these manifolds to have Property~(D) it is necessary that $\Gamma$ has subexponential growth.

One aspect of Theorem~\ref{T: main manifold thm} is that the Dixmier trace on the left-hand side is defined regardless of whether $P$ admits a DOS. Hence, the left-hand side can be interpreted as a generalisation of the density of states.

Furthermore, observe that Euclidean space is a manifold of bounded geometry satisfying Property~(D), and Schr\"odinger operators $H=-\Delta + M_V$ with smooth bounded potential $V$ are operators of the required class. Therefore, besides some minor smoothness assumptions the main result in this paper is a generalisation of~\cite{AMSZ}.

However, whereas the proofs in~\cite{AMSZ} were based on delicate singular value estimates particular to Euclidean space, our approach here is more abstract.

The general theorem that is the motor behind this paper is the following. It is a significant generalisation of~\cite[Theorem~5.7]{AMSZ}.

\begin{Theorem}\label{main_theorem}
 Let $W$ and $P$ be linear operators, such that $P$ is self-adjoint and lower-bounded and $W$ is positive and bounded. Assume that for every $t>0$, we have
 \begin{enumerate}\itemsep=0pt
 \item[$(i)$] $\exp(-tP)W \in \mathcal{L}_{1,\infty},$ %\label{cwikelcond1}
 \item[$(ii)$] $\exp(-tP)[P,W] \in \mathcal{L}_1$. %\label{cwikelcond2}
 \end{enumerate}
 Then, for every extended limit $\omega$,
 \[
 \mathrm{Tr}_{\omega}\big({\rm e}^{-tP}W\big) = \lim_{\varepsilon\to 0} \varepsilon \mathrm{Tr}\big({\rm e}^{-tP}\chi_{[\varepsilon,\infty)}(W)\big),
 \]
 whenever the limit on the right-hand side exists.
\end{Theorem}

As an application of Theorem~\ref{T: main manifold thm}, we will look at Roe's index theorem on open manifolds~\cite{Roe1988a}. Roe's index theorem is one approach of many that extends Atiyah--Singer's index theorem~\cite{AtiyahSinger1963} to non-compact manifolds, namely non-compact manifolds of bounded geometry that admit a~\textit{regular exhaustion}. For the precise definition, we refer to~\cite[Section~6]{Roe1988a}, but for this introduction it suffices to know that Property~(D) is a stronger assumption.

Given a \textit{compact} manifold $X$ with two vector bundles $E,F \to X$ and an elliptic differential operator $D\colon \Gamma(E) \to \Gamma(F)$, the local index formula equates the Fredholm index of $D$ to the integral of a differential form given by topological data denoted here simply by $\mathbf{I}(D)$,
\begin{equation}\label{eq:Atiyah-Singer}
\operatorname{Ind}(D) = \int_X \mathbf{I}(D).
\end{equation}
A special case of the index formula can be proved with the McKean--Singer formula, which in this case states that
\[
 \operatorname{Ind}(D) = \mathrm{Tr}\big(\eta {\rm e}^{-t\mathcal{D}^2}\big),
\]
where $\eta$ is the grading operator on the bundle $E\oplus F\to X,$ and $\mathcal{D}$ is the self-adjoint operator acting on $\Gamma(E\oplus F)$ by the formula
\[
 \mathcal{D} = \begin{pmatrix} 0 & D^* \\ D & 0\end{pmatrix}.
\]

If $X$ is not compact, neither side of~\eqref{eq:Atiyah-Singer} is well defined in general. In the setting of a~non-compact Riemannian manifold $X$ with a regular exhaustion, and with a graded Clifford bundle~${S \to X}$ also of bounded geometry (see Section~\ref{S:Roe}), which comes with a natural first-order elliptic differential operator, the Dirac operator $D\colon \Gamma(S_+) \to \Gamma(S_-)$, Roe modifies both sides of equation~\eqref{eq:Atiyah-Singer} as follows. Defining a linear functional $m$ on bounded $d$-forms via an averaging procedure, the right-hand side simply becomes $m(\mathbf{I}(D))$. For the left-hand side, Roe defines an algebra of uniformly smoothing operators $\mathcal{U}_{-\infty}$, and shows that elliptic operators are invertible modulo $\mathcal{U}_{-\infty}$. Recall that for operators that are invertible modulo compact operators (Fredholm operators), the Fredholm index is an element of the $K$-theory group $K_0(\mathcal{K}((\mathcal{H})) = \mathbb{Z}$~\cite{Wegge-Olsen1993}. In this case, we can similarly define an abstract index of an elliptic operator $D$ as an element of~$K_0(\mathcal{U}_{-\infty})$, by observing that $\mathcal{U}_{-\infty}$ forms an ideal in what Roe defines as uniform operators~$\mathcal{U}$. Furthermore, using the functional $m$ one can define a trace $\tau$ on $\mathcal{U}_{-\infty}$. This trace can be extended to a trace on the matrix algebras $M_n\big(\mathcal{U}_{-\infty}^+\big)$ \big(with $\mathcal{U}_{-\infty}^+$ denoting the unitisation of $\mathcal{U}_{-\infty}$\big) by putting $\tau(1)=0$ when passing to the unitisation, and then tensoring with the usual trace on $M_n(\mathbb{C})$. The tracial property gives that this map descends to a map called the dimension-homomorphism $\dim_\tau\colon K_0(\mathcal{U}_{-\infty}) \to \mathbb{R}$. These ingredients give Roe's index theorem:
\begin{equation*}
\dim_\tau(\operatorname{Ind}(D)) = m(\mathbf{I}(D)).
\end{equation*}

The nature of the averaging procedure that Roe develops is such that if $D^2$ admits a density of states, Theorem~\ref{T: main manifold thm} leads to a Dixmier trace reformulation of the analytical index $\dim_\tau(\operatorname{Ind}(D))$. In Section~\ref{S:Roe}, we prove
\[
\dim_\tau(\operatorname{Ind}(D)) = \mathrm{Tr}_\omega\big(\eta \exp\bigl(-tD^2\bigr)M_w\big), \qquad t>0,
\]
where $\eta$ is the grading on $S$.

\section{Preliminaries}\label{S: Prelims}

\subsection{Operator theory}
In this section, we will briefly summarise the operator theory used in this paper. For a more thorough exposition, we refer to~\cite{LSZVol1, LMSZVol2}.

We work on complex Hilbert spaces $\mathcal{H}$, on which the set of bounded operators is denoted by~$\mathcal{B}(\mathcal{H})$ and the ideal of compact operators by $\mathcal{K}(\mathcal{H}).$ The operator norm on $\mathcal{B}(\mathcal{H})$ is denoted~$\|\cdot\|.$ For any compact operator $T \in \mathcal{K}(\mathcal{H})$, an eigenvalue sequence $\lambda(T) = \{\lambda(k,T)\}_{k=0}^\infty$ is a~sequence of the eigenvalues of $T$ listed with multiplicity, ordered such that $ \{|\lambda(k,T)|\}_{k=0}^\infty$ is non-increasing. The singular value sequence $\mu(T) = \{\mu(k,T)\}_{k=0}^\infty$ of $T$ is defined by
\[
\mu(k,T) :=\lambda(k,|T|), \qquad k \geq 0.
\]
Equivalently,
\[
\mu(k,T) = \inf\{\|T-R\|\colon \mathrm{rank}(R)\leq k\}.
\]
Let $\ell_\infty$ denote the space of complex-valued bounded sequences indexed by $\mathbb{N}.$ For $x\in \ell_\infty$, we will denote by $\mu(x) = \{\mu(k,x)\}_{k=0}^\infty \in \ell_\infty$ the decreasing rearrangement of $ \{|x_k|\}_{k=0}^\infty.$ This is consistent with the notation for the singular value sequence of an operator in the sense that if $\mathrm{diag}(x)$ is the operator given by a diagonal matrix with entries $\{x_k\}_{k=0}^\infty,$ then $\mu(\mathrm{diag}(x)) = \mu(x).$

For $p,q\in (0,\infty)$, recall the definition of the Lorentz sequence spaces $\ell_{p,q}$, $\ell_{p,\infty}$ and $\ell_{\infty, q}$ as the spaces of sequences $x\in \ell_\infty$ such that
\begin{gather*}
\| x \|_{p,q} := \left( \sum_{k=0}^\infty (k+1)^{\frac{q}{p}-1} \mu(k,x)^q \right)^{\frac{1}{q}}<\infty,
\\
\| x \|_{p, \infty} := \sup_{k\geq 0}\, (k+1)^{\frac{1}{p}} \mu(k,x) <\infty,
\end{gather*}
and
\[
\| x \|_{\infty, q} := \left( \sum_{k=0}^\infty (k+1)^{-1} \mu(k,x)^q \right)^{\frac{1}{q}}<\infty,
\]
respectively. The space $\ell_{\infty, \infty}$ is defined as $\ell_{\infty, \infty} := \ell_\infty$, and $\ell_{p,p}$ is denoted as $\ell_p$. Using the previously defined singular value sequences, we give the definition of the Lorentz ideals $\mathcal{L}_{p,q}$ for~$p,q \in (0,\infty]$ as the quasi-normed spaces of compact operators $T$ such that
\[
\|T\|_{p,q} := \| \mu(T)\|_{\ell_{p,q}} <\infty.
\]
Like for the sequence spaces, $\mathcal{L}_{p,p}$ is denoted as $\mathcal{L}_p$, and $\mathcal{L}_{\infty, \infty} := \mathcal{B}(\mathcal{H})$. Indeed, the operator norm $\| \cdot \|$ on $\mathcal{B}(\mathcal{H})$ is sometimes denoted by $\| \cdot \|_\infty$ for clarity. The ideal $\mathcal{L}_1$ is the familiar ideal of trace-class operators (on which we can define the usual operator trace $\mathrm{Tr}$), and $\mathcal{L}_{1,\infty}$ is called the ideal of weak trace-class operators.

These quasi-norms have the property that, for $p,q \in (0,\infty],$
\[
\|ABC\|_{p,q} \leq \|A\|_{\infty} \|B\|_{p,q} \|C\|_{\infty}, \qquad B \in \mathcal{L}_{p,q}, \quad A, C \in \mathcal{B}(\mathcal{H}).
\]
We will have need of the following inequality:
\[
\|AB\|_{1} \leq 2\|A\|_{1,\infty} \|B\|_{\infty,1}, \qquad A \in \mathcal{L}_{1,\infty}, \quad B \in \mathcal{L}_{\infty, 1},
\]
which can easily be checked via the definitions of these quasi-norms and the inequality
\[
\mu(2k,AB)\leq \mu(k,A)\mu(k,B).
\]

A linear functional $\phi\colon \mathcal{L}_{1,\infty} \to \mathbb{C}$ is called a trace if
\[
\phi(TB) = \phi(BT), \qquad T\in \mathcal{L}_{1,\infty},\quad B\in \mathcal{B}(\mathcal{H}).
\]

In contrast to the situation on $\mathcal{L}_1$, continuous traces on $\mathcal{L}_{1,\infty}$ are far from unique. In particular, one important class of traces are called Dixmier traces. A linear functional $\omega$ on $\ell_{\infty}$ is called an extended limit if $\omega(x) = 0$ for all $x\in \ell_\infty$ that converge to $0$ at infinity and $\omega(\mathbf{1}) = 1,$ where $\mathbf{1} = \{1\}_{k=0}^\infty$ is the constant sequence. For any such extended limit, the formula
\[
\mathrm{Tr}_\omega(T) := \omega \Bigg( \Bigg \{ \frac{1}{\log(2+N)} \sum_{k=0}^N \lambda(k,T)\Bigg \}_{N=0}^\infty\Bigg), \qquad T\in \mathcal{L}_{1,\infty}
\]
defines a continuous trace on $\mathcal{L}_{1,\infty}$~\cite[Theorem 6.1.2]{LSZVol1}, called a Dixmier trace. From the definition, one can immediately see that $\mathcal{L}_1$ is in the kernel of any Dixmier trace.

\subsection{Left-disjoint families of operators}
In this subsection, we recall some facts about sums of left-disjoint families of operators, and prove and estimate for their $\mathcal{L}_{p}$-norm and $\mathcal{L}_{p,\infty}$-norm. A family  $ \{T_j \}_{j=0}^\infty$
 of bounded linear operators on a Hilbert space $\mathcal{H}$ is \emph{left-disjoint} if $T_j^*T_k = 0$ for all $j\neq k.$

Given a sequence $\big\{T_j\big\}_{j=0}^\infty$ of bounded linear operators on $\mathcal{H}$, let
\[
 \bigoplus_{j=0}^\infty T_j
\]
denote the operator on $\mathcal{H}\otimes \ell_2(\mathbb{N})$ given by
\[
 \sum_{j=0}^\infty T_j\otimes e_je_j^*,
\]
where $e_je_j^*$ is the rank $1$ projection onto the orthonormal basis element $e_j \in \ell_2.$

Note that
\[
 \mu\Bigg(\bigoplus_{j=0}^\infty T_j\Bigg) = \mu\Bigg(\bigoplus_{j=0}^\infty \mathrm{diag}\big(\mu\big(T_j\big)\big)\Bigg).
\]
To put it differently, the singular value sequence of the direct sum $\bigoplus_{j=0}^\infty T_j$ is the decreasing rearrangement of the sequence indexed by $\mathbb{N}^2$ given by
\[
 \big\{\mu\big(k,T_j\big)\big\}_{j,k= 0}^\infty.
\]
By definition, we have $\|T\|_{p_1,p_2} = \|\mu(T)\|_{\ell_{p_1,p_2}}.$ Therefore,
\[
 \Bigg\|\bigoplus_{j=0}^\infty T_j\Bigg\|_{p_1, p_2} = \big\|\big\{\mu\big(k,T_j\big)\big\}_{j,k\geq 0}\big\|_{\ell_{p_1, p_2}(\mathbb{N}^2)}.
\]
Now let $q>0$. We have for each $j$ that  $\mu(k,T_j)\leq (k+1)^{-\frac1q}\|T_j\|_{q,\infty}$.  Therefore,
\[
 \Bigg\|\bigoplus_{j=0}^\infty T_j\Bigg\|_{p_1,p_2} \leq \big\|\big\{(1+k)^{-\frac1q}\|T_j\|_{q,\infty}\big\}_{j,k\geq 0}\big\|_{\ell_{p_1,p_2}(\mathbb{N}^2)}.
\]
Lemma 4.3 of~\cite{LevitinaSukochev2020} implies that if $q$ is sufficiently small, then there exists a constant $C_{p_1,p_2,q}$ such that for any sequence $\{x_j\}_{j=0}^\infty$, we have
\[
 \bigl\|\bigl\{(1+k)^{-\frac1q}x_j\bigr\}_{j,k\geq 0}\bigr\|_{\ell_{p_1,p_2}(\mathbb{N}^2)} \leq C_{p_1,p_2,q}\|\{x_j\}_{j=0}^\infty\|_{\ell_{p_1,p_2}}.
\]
Taking $x_j=\|T_j\|_{q,\infty}$ and sufficiently small $q$ (depending on $p_1$, $p_2$), we have
\begin{equation}\label{direct_sum_quasinorm_bound}
 \Bigg\|\bigoplus_{j=0}^\infty T_j\Bigg\|_{p_1,p_2} \leq C_{p_1,p_2,q} \big\|\{\|T_j\|_{q,\infty}\}_{j=0}^\infty \big\|_{p_1,p_2}.
\end{equation}

Combining~\cite[Proposition 2.7 and Lemma 2.9]{LevitinaSukochev2020} gives the following result.
\begin{Lemma}\label{disjointification}
 Let $0<p<2,$ and let $\big\{T_j\big\}_{j=0}^\infty$ be a left-disjoint family of bounded linear operators. There exist constants $C_p$, $C_p'$ such that
 \begin{align*}
& \Bigg\|\sum_{j=0}^\infty T_j\Bigg\|_{p,\infty}\leq C_p\Bigg\|\bigoplus_{j=0}^\infty T_j\Bigg\|_{p,\infty},\\
& \Bigg\|\sum_{j=0}^\infty T_j\Bigg\|_{p}\leq C_p'\Bigg\|\bigoplus_{j=0}^\infty T_j\Bigg\|_{p}.
 \end{align*}
\end{Lemma}

A combination of~\eqref{direct_sum_quasinorm_bound} and Lemma~\ref{disjointification} immediately yields the following:
\begin{Corollary}\label{disjointification_corollary}
 Let $0<p<2,$ and let $\big\{T_j\big\}_{j=0}^\infty$ be a left-disjoint family of operators. There exist $q, q'>0$ $($depending on $p)$ and constants $C_{p,q}, C'_{p,q'}>0$ such that
 \begin{align*}
& \Bigg\|\sum_{j=0}^\infty T_j\Bigg\|_{p,\infty}\leq C_{p,q}\big\|\{\|T_j\|_{q,\infty}\}_{j=0}^\infty\big\|_{\ell_{p,\infty}},\\
& \Bigg\|\sum_{j=0}^\infty T_j\Bigg\|_{p}\leq C'_{p,q'}\big\|\{\|T_j\|_{q',\infty}\}_{j=0}^\infty\big\|_{\ell_{p}}.
 \end{align*}
\end{Corollary}

\begin{Remark}
 For $p=1$, we can take any $0<q,q'<1$.
\end{Remark}

\subsection{Preliminaries on manifolds}
\label{S: Manifolds}
As mentioned previously, all manifolds in this paper
are smooth, non-compact and connected.

\begin{Definition}
Let $(X,g)$ be a Riemannian manifold. The injectivity radius $i(x)$ at a point $x \in X$ is defined as
\[
i(x):= \sup \{ R \in \mathbb{R}_{\geq 0} \colon \exp_x|_{B(0,R)} \text{ is a diffeomorphism}\},
\]
where $\exp_x$ is the exponential map at $x$ and $B(0,R) \subseteq T_xX$ is the metric ball with radius $R$ centered around the origin. The injectivity radius of the manifold $X$ is defined as \[i_g := \inf_{x \in X} i(x).\]
\end{Definition}

It is a theorem that the injectivity radius map \begin{align*}
 i \colon \ X&\to (0, \infty],\\
 x&\mapsto i(x)
\end{align*}
is continuous~\cite[Proposition~III.4.13]{Sakai1996}.

\begin{Definition}\label{D:bounded_geometry}
A Riemannian manifold $(X,g)$ has \textit{bounded geometry} if the injectivity radius~$i_g$ satisfies
$i_g >0$,  and the Riemannian curvature tensor $R$ and all its covariant derivatives are uniformly bounded.
\end{Definition}

This is the definition as in~\cite[Definition 1.1]{Kordyukov1991},~\cite[Chapter II]{Eichhorn2008} and~\cite[Chapter 3]{Eichhorn2007}. Every open manifold admits a metric of bounded geometry~\cite{Greene1978}.

As is well known (see, e.g.,~\cite[Lemma 2.4]{Kordyukov1991}), bounded geometry implies that there exists~${r_0>0}$ and a countable set $\{x_j\}_{j=0}^\infty$ of points in $X$ such that
\begin{enumerate}\itemsep=0pt
 \item{} $X = \bigcup_{j=0}^\infty B(x_j,r_0)$.
 \item{} Each ball $B(x_j,r_0)$ is a chart for the exponential normal coordinates based at $x_j$.
 \item{} The covering $\{B(x_j,r_0)\}_{j=0}^\infty$ has finite order, meaning that there exists $N$ such that each ball intersects at most $N$ other balls.
 \item{} $\sup_j \mathrm{Vol}(B(x_j,r_0)) < \infty.$
 \item{} There exists a partition of unity $\{\psi_j\}_{j=0}^\infty$ subordinate to $\{B(x_j,r_0)\}$ such that for every~$\alpha$, we have $\sup_{j,x} |\partial^{\alpha}\psi_j(x)|<\infty,$ where $\partial^{\alpha}$ is taken in the exponential normal coordinates of~$B(x_j,r_0).$
\end{enumerate}
Without loss of generality, $r_0$ can be taken smaller or equal to $1$ (see~\cite[Lemma~2.3]{Kordyukov1991}).

We will briefly refer to the scale of Hilbert Sobolev spaces $\{H^s(X)\}_{s\in \mathbb{R}}$ defined over $X$ in \cite[Section 3]{Kordyukov1991} and which have several equivalent characterisations, see, e.g.,~\cite{GrosseSchneider2013}. The most important characterisation for us is that if $P \in BD^m(X),$ then $P$ defines a bounded linear operator from $H^{s+m}(X)$ to $H^s(X)$ for every $s\in \mathbb{R}$.

\begin{Remark}
 It follows from the bounded geometry assumption that there exist constants $c,C>0$ such that for every $x\in X$ and $R>0$, we have
 \[
 |B(x,R)| \leq C\exp(cR).
 \]
\end{Remark}

Note that for almost every $r>0$, we have
\[
 \frac{{\rm d}}{{\rm d}r}|B(x,r)| = |\partial B(x,r)|.
\]
See~\cite[Propositions III.3.2 and~III.5.1]{Chavel2006}.

A standard example of a manifold of bounded geometry is a covering space of a compact manifold. The following is well known, but lacking a reference we supply a proof for the reader's convenience.
\begin{Proposition}
Let $X$ be a complete $d$-dimensional Riemannian manifold with metric $g$. Let $\Gamma$ be a discrete, finitely generated subgroup of the isometries of $(X,g)$ which acts freely and properly discontinuously on $X$ such that the quotient $X/\Gamma$ is a compact $(d$-dimensional$)$ Riemannian manifold. Then $X$ has bounded geometry.
\begin{proof}
Since $\Gamma$ acts cocompactly on $X$, there exists a compact set $L \!\subseteq\! X$ such that  ${\bigcup_{\gamma \in \Gamma} \gamma L \!=\! X}$. Indeed, since the open balls $B(x, 1)$, $x\in X$ project onto an open cover of $X/\Gamma$, there exists a~finite collection $\{ B(x_i,1)\}_{i=1}^N$ that projects onto $X/\Gamma$. Defining $L := \bigcup_{i=1}^N \overline{B(x_i,1)}$ gives the claimed compact set $L$.

Set \[i_L := \inf_{x \in L} i(x), \]
where $i(x)$ is the injectivity radius at the point $x$. Since the injectivity radius is a continuous function on $X$, $i(x)>0$ for all $x\in X$, and $L$ is compact, it follows that $i_L > 0$. Since $\Gamma$ acts by isometries, for any $\gamma \in \Gamma$, we have $i(\gamma x) = i(x)$, see, for example, the proof of~\cite[Theorem~III.5.4]{Sakai1996}. Therefore, $\inf_{x\in X} i(x) = i_L >0$.

Next, since the curvature tensor $R$ is smooth, clearly $R$ and all its covariant derivatives are bounded on $L$. Let $\Phi\colon X \to X$ be the isometry by which $\gamma \in \Gamma$ acts. Then for all $x\in X$ $\Phi_x^*\colon T_xX \to T_{\Phi(x)}X$ is an isomorphism, and by~\cite[p.~41]{Sakai1996},
\begin{align*}
 &\Phi_x^*(\nabla_{\xi}\eta)= \nabla_{\Phi_x^*\xi} \Phi_x^*\eta,\\
 &\Phi_x^*(R(\eta, \xi)\zeta)= R(\Phi_x^*\eta, \Phi_x^*\xi)\Phi_x^*\zeta,
\end{align*}
where $\xi, \eta, \zeta \in T_xX$. Combine the facts that $R$ and its covariant derivatives are bounded on~$L$, that $\bigcup_{\gamma \in \Gamma} \gamma L = X$, and that isometries preserve $R$ and taking covariant derivatives in the above manner, and we can conclude that $R$ and its covariant derivatives are uniformly bounded on~$X$.
\end{proof}
\end{Proposition}

\begin{Remark}\label{R: metric balls lower bound}
A non-compact Riemannian manifold of bounded geometry has infinite volume. This can be checked easily via the covering $X = \bigcup_{j=0}^\infty B(x_j,r_0)$ below Definition~\ref{D:bounded_geometry}, and the observation that $\inf_{j} |B(x_j, r_0)| > 0$~\cite{Kodani1988}.
\end{Remark}

\section{Proof of Theorem~\ref{main_theorem}}
%\label{S: Main thm}
Part of the proof in~\cite{AMSZ} used the identity
\begin{equation}\label{eq: Abelian Rd}
 \lim_{s \downarrow 1} (s-1)\int_{\mathbb{R}^d} F(x)\big(1+|x|^2\big)^{-\frac{s}{2}}\,{\rm d}x = \lim_{R\to\infty} R^{-d}\int_{B(0,R)}F(x)\,{\rm d}x
\end{equation}
for any bounded measurable function $F$ on $\mathbb{R}^d$ such that the right-hand side exists.

Equation~\eqref{eq: Abelian Rd} can be generalised in the following manner. The proof is essentially the same as~\cite[Lemma 6.1]{AMSZ}.

\begin{Lemma}\label{L: trace formula}
 Let $A$ and $B$ be bounded linear operators, $B\geq 0$, such that $AB^s \in \mathcal{L}_1$ for every $s>1.$ Then
 \[
 \lim_{s\downarrow 1}(s-1)\mathrm{Tr}(AB^s) = \lim_{\varepsilon\to 0} \varepsilon \mathrm{Tr}(A\chi_{[\varepsilon,\infty)}(B))
 \]
 whenever the limit on the right exists.
\end{Lemma}
\begin{proof}
 Writing $B^s = \int_0^{\|B\|_{\infty}} \lambda^s\,{\rm d}E_B(\lambda)$, $\lambda^s = s\int_0^\lambda r^{s-1}\,{\rm d}r$
 and applying Fubini's theorem yields
 \[
 \mathrm{Tr}(AB^s) = s\int_0^{\|B\|_{\infty}} r^{s-1} \mathrm{Tr}(A\chi_{[r,\infty)}(B))\,{\rm d}r.
 \]
 Assume without loss of generality that $\|B\|_{\infty} = 1.$
 Writing $F(r) := \mathrm{Tr}(A\chi_{[r,\infty)}(B))$, our assumption is that
 \[
 F(r) \sim \frac{c}{r},\qquad r\to 0.
 \]
 and
 \[
 \mathrm{Tr}(AB^s) = s\int_0^1 r^{s-1}F(r)\,{\rm d}r.
 \]
 Let $\varepsilon>0,$ and choose $R>0$ sufficiently small such that if $0<r<R$, then
 \[
 |rF(r)-c|<\varepsilon.
 \]
 We write $\mathrm{Tr}(AB^s)-\frac{c}{s-1}$ as
 \[
 \mathrm{Tr}(AB^s)-\frac{c}{s-1} = c+s\int_0^1 \bigl(r^{s-1}F(r) - c r^{s-2}\bigr)\,{\rm d}r.
 \]
 Therefore,
 \begin{align*}
 \bigg|\mathrm{Tr}(AB^s)-\frac{c}{s-1}\bigg| & \leq |c|+s\int_0^1 r^{s-2}|rF(r)-c|\,{\rm d}r\\
 & \leq |c|+s\int_R^1 r^{s-2}|rF(r)-c|\,{\rm d}r + \frac{s\varepsilon}{s-1}.
 \end{align*}
 It follows that
 \[
 |(s-1)\mathrm{Tr}(AB^s)-c| = O(s-1)+s\varepsilon,\qquad s\downarrow 1.
 \]
 Since $\varepsilon$ is arbitrary, this completes the proof.
\end{proof}

We will make use of the following theorem, which is~\cite[Theorem~1.3.20]{LMSZVol2}.
\begin{Theorem}\label{LSZ2ethm}
 Let $A$ and $B$ be positive bounded linear operators such that $\big[B,A^{\frac12}\big] \in \mathcal{L}_1$ and $AB\in \mathcal{L}_{1,\infty}$. If
 \[
 \big\|A^{\frac12}B^s\big\|_1 = o\big((s-1)^{-2}\big),\qquad s\downarrow 1,
 \]
 then for every extended limit $\omega$, we have
 \[
 \mathrm{Tr}_{\omega}(AB) = \lim_{s\downarrow 1} (s-1)\mathrm{Tr}(AB^s)
 \]
 if the limit exists.
\end{Theorem}
Hence, if $\big[B,A^{\frac12}\big] \in \mathcal{L}_1$ and $\big\|A^{\frac12}B^s\big\|_1 = o\big((s-1)^{-2}\big),$ then by Lemma~\ref{L: trace formula}
\[
 \mathrm{Tr}_{\omega}(AB) = \lim_{\varepsilon\to 0} \varepsilon\mathrm{Tr}(A\chi_{[\varepsilon,\infty)}(B))
\]
whenever the limit on the right exists.

Recall the Araki--Lieb--Thirring inequality~\cite{Kosaki1992}
\begin{equation}\label{alt_inequality}
 \|AB\|_{r,\infty}^r \leq {\rm e}\|A^rB^r\|_{1,\infty},\qquad r > 1
\end{equation}
and the numerical inequality
\begin{equation}\label{zeta_inequality}
 \|X\|_{\infty,1} = \sum_{k=0}^\infty \frac{\mu(k,X)}{k+1} \leq \|X\|_{q,\infty}\zeta\bigg(1+\frac1q\bigg),
\end{equation}
obtained by simply writing out the definitions. Here $\zeta$ is the Riemann zeta function, and it obeys $\zeta\big(1+\frac1q\big) \sim q$ as $q\to\infty.$

\begin{Corollary}%\label{C:main_result}
 Let $A$ and $B$ be positive bounded linear operators such that
 \begin{enumerate}\itemsep=0pt
 \item[$(i)$] $\big[A^{\frac12},B\big] \in \mathcal{L}_1,$ %\label{cond1}
 \item[$(ii)$] $A^{\frac14}B \in \mathcal{L}_{1,\infty}.$ %\label{cond2}
 \end{enumerate}
 Then for every extended limit $\omega$, we have
 \[
 \mathrm{Tr}_{\omega}(AB) = \lim_{s\downarrow 1} (s-1)\mathrm{Tr}(AB^s)
 \]
 if the limit exists.
\end{Corollary}

\begin{proof}
 Let $1<s<2.$ We have
 \begin{align*}
 \big\|A^{\frac12}B^s\big\|_1 &{}\leq \big\|\big[A^{\frac12},B\big]B^{s-1}\big\|_1+\big\|BA^{\frac12}B^{s-1}\big\|_1 \\
 &{}\leq \big\|\big[A^{\frac12},B\big]\big\|_1\|B\|_{\infty}^{s-1}+\big\|BA^{\frac14}\big\|_{1,\infty}\big\|A^{\frac14}B^{s-1}\big\|_{\infty,1}.
 \end{align*}
 By the numerical inequality~\eqref{zeta_inequality} above,
 \[
 \big\|A^{\frac14}B^{s-1}\big\|_{\infty,1} \lesssim (s-1)^{-1}\big\|A^{\frac14}B^{s-1}\big\|_{\frac{1}{s-1},\infty}.
 \]
 Since $s<2$, we have $\frac{1}{s-1}>1$, and hence the Araki--Lieb--Thirring inequality~\eqref{alt_inequality} delivers
 \[
 \big\|A^{\frac14}B^{s-1}\big\|_{\frac{1}{s-1},\infty}^{\frac{1}{s-1}} \leq {\rm e}\big\|A^{\frac{1}{4(s-1)}}B\big\|_{1,\infty}\leq {\rm e}\big\|A^{\frac14}\big\|_{\infty}^{\frac{2-s}{s-1}}\big\|A^{\frac14}B\big\|_{1,\infty}.
 \]
 This verifies the assumptions of Theorem~\ref{LSZ2ethm}.
\end{proof}

In our case, we have an operator $P$ which is self-adjoint and lower-bounded, which means $\exp(-tP)$ is positive and bounded for all $t>0$, and we take $B=W$. Then the assumptions become
\begin{equation*}%\label{eq:original_conditions}
 [\exp(-tP),W] \in \mathcal{L}_1,\qquad \exp(-tP)W \in \mathcal{L}_{1,\infty}
\end{equation*}
for every $t>0.$

The former condition can be modified to one which is easier to verify in geometric examples.
\begin{Lemma}[Duhamel's formula]\label{Duhamel formula}
Let $P$ be a lower-bounded self-adjoint operator on a Hilbert space $H$, and let $W$ be a bounded operator. Then
\[
[\exp(-tP),W] = -\int_0^t \exp(-sP)[P,W]\exp(-(t-s)P)\,{\rm d}s.
\]
\begin{proof}
The method of proof is standard, see for example~\cite[Lemma~5.2]{AzamovCarey2009}. Define $F\colon [0,t] \to \mathcal{L}(H)$ by $F(s) = \exp(-sP)W\exp(-(t-s)P)$. Since $P$ is lower-bounded, $\exp(-sP)$ is bounded for all $s \in [0,t]$. Hence the derivative of $F(s)$ in the strong operator topology is
\begin{align*}
F'(s) &= -P\exp(-sP)W\exp(-(t-s)P) + \exp(-sP)W P \exp(-(t-s)P) \\
&= -\exp(-sP)[P,W]\exp(-(t-s)P),
\end{align*} in the sense that \[\lim_{h\to 0} \frac{1}{h}(F(s+h)-F(s))\xi = F'(s)\xi, \qquad \xi \in H.\]
Therefore, by the fundamental theorem of calculus for Banach space-valued functions, we can conclude that for $\xi \in H$, \begin{align*}
[\exp(-tP),W]\xi &= (F(t)-F(0))\xi= \int_{0}^t F'(s) \xi \,{\rm d}s\\
&= -\int_0^t \exp(-sP)[P,W]\exp(-(t-s)P) \xi \,{\rm d}s.
\tag*{\qed}
\end{align*}
\renewcommand{\qed}{}
\end{proof}
\end{Lemma}

\begin{Lemma}
 Let $P$ be a self-adjoint lower bounded linear operator
 and let $W$ be a bounded self-adjoint operator such that $[P,W]$ makes sense, and
 \[
 \exp(-tP)[P,W] \in \mathcal{L}_1,\qquad t > 0.
 \]
 Then
 \[
 [\exp(-tP),W]\in \mathcal{L}_1,\qquad t > 0.
 \]
\end{Lemma}
\begin{proof}
 By the Duhamel formula (Lemma~\ref{Duhamel formula}),
 \begin{align*}
 [\exp(-tP),W]={}& -\int_0^{\frac{t}{2}}\exp(-sP)[P,W]\exp(-(t-s)P)\,{\rm d}s\\
 & -\int_{\frac{t}{2}}^t \exp(-sP)[P,W]\exp(-(t-s)P)\,{\rm d}s.
 \end{align*}
 For $0<s<\frac{t}{2}$, we have
 \[
 \|\exp(-sP)[P,W]\exp(-(t-s)P)\|_1 \leq \bigg\|[P,W]\exp\biggl(-\frac{t}{2}P\biggr)\bigg\|_1
 \]
 while for $\frac{t}{2}<s<t,$
 \[
 \|\exp(-sP)[P,W]\exp(-(t-s)P)\|_1 \leq \bigg\|\exp\biggl(-\frac{t}{2}\biggr)[P,W]\bigg\|_1.
 \]
 Hence, the triangle inequality for weak integrals, we have
 \[
 \|[\exp(-tP),W]\|_1 \leq t\|[P,W]\|_1.
 \tag*{\qed}
 \]
 \renewcommand{\qed}{}
\end{proof}

Combining the results of this section yields the following.
\begin{Theorem}%\label{abstract_AMSZ_theorem}
 Let $P$ be a self-adjoint lower-bounded linear operator, and $W$ a positive bounded linear operator such that
 \[
 \exp(-tP)W\in \mathcal{L}_{1,\infty},\qquad \exp(-tP)[P,W] \in \mathcal{L}_{1}
 \]
 for every $t>0.$
 Then for every extended limit $\omega$, we have
 \[
 \mathrm{Tr}_{\omega}(\exp(-tP)W) = \lim_{s\downarrow1} (s-1)\mathrm{Tr}(\exp(-tP)W^s)
 \]
 conditional on the existence of the right-hand side.
\end{Theorem}

More to the point, we have
\[
 \mathrm{Tr}_{\omega}(\exp(-tP)W) = \lim_{\varepsilon\to 0} \varepsilon\mathrm{Tr}(\exp(-tP)\chi_{[\varepsilon,\infty)}(W))
\]
whenever the right-hand side limit exists. This completes the proof of Theorem~\ref{main_theorem}.

\section{Manifolds of bounded geometry}%\label{S:Proof_Manifolds}
Let us now shift our attention to the case where we have a Riemannian manifold of bounded geometry $X$ and a self-adjoint, lower-bounded operator $P \in EBD^2(X)$.

Estimates of the form
\[
 \|M_fg(-{\rm i}\nabla)\|_{p,\infty} \leq c_p\|f\|_p\|g\|_{p,\infty}
\]
for $p>2$ are sometimes called Cwikel-type estimates. Here, $f$ and $g$ are function on $\mathbb{R}^d,$
see \cite[Chapter 4]{Simon2005}. Similar estimates with $p<2$ were obtained earlier by Birman and Solomyak \cite{BirmanSolomyak1969}.
In particular, it follows from~\cite[Theorem~4.5]{Simon2005} that if $f$ is a measurable function on $\mathbb{R}^d$ then for every $t>0$ and $0<p<2$ that we have
\[
 \big\|M_f{\rm e}^{t\Delta}\big\|_{p} \lesssim_{t,p} \bigg(\sum_{k\in \mathbb{Z}^d} \|f\|_{L_\infty(k+[0,1)^d)}^p\bigg)^{\frac1p}
\]
and similarly
\[
 \big\|M_f{\rm e}^{t\Delta}\big\|_{p,\infty} \lesssim_{t,p} \big\|\{\|f\|_{L_{\infty}(k+[0,1)^d)}\big\|_{p,\infty}.
\]
We seek similar estimates for functions $f$ on manifolds of bounded geometry. In place of the decomposition of $\mathbb{R}^d$ into cubes, $\mathbb{R}^d = \bigcup_{k\in \mathbb{Z}}^d [0,1)^d+k$, we will use the partition of unity constructed according to Section~\ref{S: Manifolds}.

Namely, we will show that for $0<p<2$, we have $\exp(-tP)M_f \in \mathcal{L}_{p,\infty}$ whenever $f\in \ell_{p,\infty}(L_\infty)$, and $\exp(-tP)M_f \in \mathcal{L}_{p}$ whenever $f\in \ell_{p}(L_\infty),$
where $\ell_{p,\infty}(L_\infty)$ and $\ell_p(L_{\infty})$ are certain function spaces on $X.$

\begin{Definition}
 Let $\{x_j\}_{j=0}^\infty$ and $r_0$ be as in Section~\ref{S: Manifolds}. Given a bounded measurable function~$f$ on~$X$, define
 \[
 \|f\|_{\ell_{p,\infty}(L_\infty)} := \big\| \{\|f\|_{L_{\infty}(B(x_j,r_0))}\}_{j=0}^\infty\big\|_{\ell_{p,\infty}}
 \]
 and
 \[
 \|f\|_{\ell_{p}(L_\infty)} := \big\| \{\|f\|_{L_{\infty}(B(x_j,r_0))}\}_{j=0}^\infty\big\|_{\ell_{p}}.
 \]
\end{Definition}

Let $B$ be an open ball in $\mathbb{R}^d$ such that at every point $x\in X$, we have a normal coordinate system $\exp_x \circ e\colon B \to B(x, r_0)$, where $e$ is the identification of $\mathbb{R}^d$ with the tangent space $T_xX$ and $\exp_x$ is the Riemannian exponential map~\cite[Proposition~1.2]{Kordyukov1991}. Via these maps, for each $x\in X$, we can pullback the metric $g$ restricted to $B(x, r_0)$ to a metric $g^x$ on $B$.

\begin{Definition}%\label{d:Riem}
A Riemannian metric $g$ on $B$ can be represented uniquely by the $d^2$ smooth functions
\[
g_{ij}:= g(\partial_i, \partial_j) \in C^\infty(B).
\]
Define $\mathrm{Riem}_b(B)$ as the set of Riemannian metrics for which the functions $g_{ij}$ extend to smooth functions in $C^\infty\bigl(\overline{B}\bigr)$. We equip $\mathrm{Riem}_b(B)$ with the topology induced by the embedding $\mathrm{Riem}_b(B)\allowbreak \subseteq (C^\infty\bigl(\overline{B}\bigr))^{d^2}$, where we take the usual topology on $C^\infty\bigl(\overline{B}\bigr)$ defined by the seminorms
\[
p_N(f) := \max_{x\in \overline{B}} \{|D^\alpha f(x)|\colon |\alpha|\leq N \}.
\]
\end{Definition}

The following is essentially~\cite{Eichhorn1991}, see in particular the discussion below Theorem A. See also the related statement~\cite[Proposition~2.4]{Roe1988a}.
\begin{Proposition}\label{p:bounded_metric}
Let $X$ be a manifold of bounded geometry. The functions $g^x_{ij}$ considered as a~family of smooth functions parametrized by $i$, $j$ and by a point $x\in X$, can be extended to $\overline{B}$ and then lie in a bounded subset of $C^\infty\bigl(\overline{B}\bigr).$
\end{Proposition}
In principle, the identification $e$ of $\mathbb{R}^d$ with the tangent space $T_xX$ can vary from point to point, and therefore
the matrix elements $g^x_{ij}$ are not uniquely defined. However this does not change the fact that for any given identification, the result of Proposition~\ref{p:bounded_metric} holds.

\begin{Corollary}\label{c:compactmetrics}
Let $(X, g)$ be a manifold with bounded geometry, and let $\{B(x_j, r_0)\}_{j \in \mathbb{N}}$ be an open cover of $X$ as in Section~$\ref{S: Manifolds}$. The set
\[
\{ g^{x_j} \colon j \in \mathbb{N}\} \subset \mathrm{Riem}_b(B)
\]
is a pre-compact subset of $C^\infty\bigl(\overline{B}\bigr).$
\end{Corollary}
\begin{proof}
By Proposition~\ref{p:bounded_metric}, the functions $g^{x_k}_{ij}$ lie in a bounded set in $C^\infty\bigl(\overline{B}\bigr)$. Since the closure of a bounded set is also bounded~\cite[Section~IV.2]{Conway1990}, and since $C^\infty\bigl(\overline{B}\bigr)$ has the Heine--Borel property~\cite[Section~1.9]{GrandpaRudin}, the assertion follows.
\end{proof}

Given $g \in \mathrm{Riem}_b(B)$, we denote $\Delta_g^D$ the self-adjoint realisation of $\Delta_g$ on $B$ with Dirichlet boundary conditions. Explicitly, $\Delta_g^D$ is defined as the operator associated to the closure of the quadratic form
\[
 q_g(u,v) = \int_{B} \sqrt{|g(x)|}\sum_{i,j} g^{i,j}(x)\partial_i u(x)\overline{ \partial_j v(x)}\,{\rm d}x,\qquad u,v\in C^\infty_c(B).
\]
\begin{Lemma}\label{l:uniformlaplacianbound}
Let $x_j$ and $g^{x_j} \in \mathrm{Riem}_b(B)$ be as in Corollary~$\ref{c:compactmetrics}$.
We have
\[
\sup_{j \in \mathbb{N}} \bigl\|\bigl(1-\Delta_{g^{x_j}}^D\bigr)^{-1}\bigr\|_{\mathcal{L}_{\frac{d}{2},\infty}(L_2(B)) }<\infty.
\]
\end{Lemma}
\begin{proof}
Let $g$ be a metric on the closed unit ball $\overline{B},$ and let $c_g$ and $c_G$ be positive constants such that
\[
 c_g\left(\sum_{k=1}^d |\xi_k|^2\right) \leq \sum_{k,l=1}^d \sqrt{|g(x)|}g^{k,l}(x)\xi_k\overline{\xi_l} \leq C_g\left(\sum_{k=1}^d |\xi_k|^2\right)
\]
for all $\xi \in \mathbb{C}^d.$
We will prove that there is a constant $k_d$ such that
\begin{equation}\label{e:metric_uniformity}
\big\|\bigl(1-\Delta_{g}^D\bigr)^{-1}\big\|_{\mathcal{L}_{\frac{d}{2},\infty}(L_2(B)) } \leq k_dc_g^{-\frac12}.
\end{equation}
Since
\[
 \inf_{j} c_{g^{x_j}} > 0,
\]
\eqref{e:metric_uniformity} implies the result.

Let $q_0$ denote the Dirichlet quadratic form on $B.$ That is,
\[
 q_0(u,v) := \sum_{j} \int_B \partial_j u(x)\overline{\partial_j v}(x)\,{\rm d}x,\qquad u,v \in C^\infty_c(B).
\]
The Dirichlet Laplacian $\Delta_0^D$ on $B$ is defined as the operator associated to the closure of the quadratic form $q_0$ (see, e.g.,~\cite[Example~7.5.26]{Simon-course-IV}).
By the definitions of $q_g,$ $c_g$ and $C_g$, we have
\[
 c_gq_0(u,u) \leq q_g(u,u) \leq C_gq_0(u,u),\qquad u\in C^\infty_c(G).
\]
It follows that the form domains of the closures of $q_0$ and $q_g$ coincide, we denote this space~$H^1_0(B).$
The preceding inequality implies, in particular, that
\[
 c_g\big\|\big(1-\Delta_0^D\big)^{\frac12}u\big\|_{L_2(B)}^2 \leq \big\|\big(1-\Delta_g^D\big)^{\frac12}u\big\|_{L_2(B)}^2,\qquad u\in H^1_0(B).
\]
By standard results in quadratic form theory (see, e.g.,~\cite[equation (7.5.29)]{Simon-course-IV}), $1-\Delta_g^D$ defines a~topological linear isomorphism from $H^1_0(B)$ to its dual \smash{$\big(H^1_0(B)\big)^{^*}$}. Therefore, replacing $u$ with $\big(1-\Delta_g^D\big)^{-1}v$ for $v\in \big(H^1_0\big)^*$, we arrive at
\[
 c_g\big\|\big(1-\Delta_0^D\big)^{\frac12}\big(1-\Delta_g^D\big)^{-1}v\big\|_{L_2(B)}^2 \leq \big\|\big(1-\Delta_g^D\big)^{-\frac12}v\big\|_{L_2(B)}^2,\qquad v \in \big(H^1_0(B)\big)^*.
\]
Replacing $v$ with $\big(1-\Delta_0^D\big)^{\frac12}w$ for $w \in L_2(B)$ gives
\[
 \big\|\big(1-\Delta_0^D\big)^{\frac12}\big(1-\Delta_g^D\big)^{-1}\big(1-\Delta_0\big)^{\frac12}w\big\|_{L_2(B)}^2 \leq c_g^{-1}\|w\|_{L_2(B)}^2.
\]
Recall that $\big\|\big(1-\Delta_0^D\big)^{-\frac12}\big\|_{\mathcal{L}_{d,\infty}(L_2(B)) }<\infty$ by standard Weyl asymptotics. Therefore,
\begin{align*}
 \big\|\big(1-\Delta_g^D\big)^{-1}\big\|_{\mathcal{L}_{\frac{d}{2},\infty}(L_2(B))}\leq{}& \big\|\big(1-\Delta_0^D\big)^{-\frac12}\big\|_{\mathcal{L}_{d,\infty}(L_2(B)) }^2 \\
 &{}\times \big\|\big(1-\Delta_0^D\big)^{\frac12}\big(1-\Delta_g^D\big)^{-1}\big(1-\Delta_0\big)^{\frac12}\big\|_{\mathcal{B}(L_2(B))}\\
\leq{}& c_g^{-\frac12}\big\|\big(1-\Delta_0^D\big)^{-\frac12}\big\|_{\mathcal{L}_{d,\infty}(L_2(B)) }^2.
\end{align*}
Defining $k_d = \bigl\|\bigl(1-\Delta_0^D\bigr)^{-\frac12}\bigr\|_{\mathcal{L}_{d,\infty}(L_2(B)) }^2$
completes the proof of~\eqref{e:metric_uniformity},
and hence of the lemma.
\end{proof}

\begin{Proposition}
For all $q>0$, there exists $M>0$ such that
\[
 \sup_{j} \big\|M_{\psi_j}(1-\Delta)^{-M}\big\|_{q,\infty} < \infty,
\]
where $\psi_j$ is the partition of unity subordinate to $\{B(x_j, r_0)\}_{j \in \mathbb{N}}$ mentioned in Section~$\ref{S: Manifolds}$.
\end{Proposition}
\begin{proof}
The proof is inspired by the proof of~\cite[Lemma~3.4.8]{SukochevZanin2018}. Let $V_j\colon L_2(X) \to L_2(B, g^{x_j})$ be the partial isometry mapping $\xi \in L_2(X)$ to $V\xi (z) = f|_{B(x_j, r_0)} \circ \exp_x \circ e \in L_2(B, g^{x_j})$. Denote $V_j\psi_j := \phi_j$. Then, by construction,
\begin{gather*}
M_{\psi_j} = V_j^* M_{\phi_j} V_j,
\\
(1-\Delta)^M M_{\psi_j} = V_j^* \big(1-\Delta_{g^{x_j}}\big)^M M_{\phi_j} V_j.
\end{gather*}
It follows that
\begin{align*}
 (1-\Delta)^M M_{\psi_j} V_j^* \big(1-\Delta_{g^{x_j}}^D\big)^{-M} V_j &{}= V_j^* \big(1-\Delta_{g^{x_j}}\big)^M M_{\phi_j} \big(1-\Delta_{g^{x_j}}^D\big)^{-M} V_j\\
 &{}=V_j^* \big[\big(1-\Delta_{g^{x_j}}\big)^M, M_{\phi_j}\big] \big(1-\Delta_{g^{x_j}}^D\big)^{-M} V_j + V_j^* M_{\phi_j} V_j\\
 &{}= \big[(1-\Delta)^M, M_{\psi_j}\big] V_j^* \big(1-\Delta_{g^{x_j}}^D\big)^{-M} V_j + M_{\psi_j}.
\end{align*}
Multiplying both sides by $(1-\Delta)^{-M}$ and rearranging gives
\begin{gather*}
 (1-\Delta)^{-M}M_{\psi_j} = M_{\psi_j} V_j^* \big(1-\Delta_{g^{x_j}}^D\big)^{-M} V_j \\
 \hphantom{(1-\Delta)^{-M}M_{\psi_j} =}{}
 - (1-\Delta)^{-M}\big[(1-\Delta)^M, M_{\psi_j}\big] V_j^* \big(1-\Delta_{g^{x_j}}^D\big)^{-M} V_j.
\end{gather*}
We claim that
\[
\sup_{j \in \mathbb{N}} \big\| (1-\Delta)^{-M}\big[(1-\Delta)^M, M_{\psi_j}\big] \big\|_\infty < \infty.
\]
For every $\alpha$, we have $\sup_{j,x} |\partial^{\alpha}\psi_j(x)|<\infty,$ where $\partial^{\alpha}$ is taken in the exponential normal coordinates of $B(x_j,r_0),$ and therefore $\big[(1-\Delta)^M, M_{\psi_j}\big]$ is a uniform differential operator of order $2M-1$ with coefficients that are uniform in $j$. Using~\cite[Theorem~3.9]{Kordyukov1991}, it follows that $\big[(1-\Delta)^M, M_{\psi_j}\big]$ is a bounded operator from $L_2(X)$ to the Sobolev space $H_{1-2M}(X)$ with norm uniform in $j$. By~\cite[Proposition~4.4]{Kordyukov1991}, $(1-\Delta)^{-M}$ is a bounded operator from that space into $L_2(X)$, and so the claim holds.

Since the norm of $V_j$ is equal to one, via Lemma~\ref{l:uniformlaplacianbound}, we get for $M$ large enough
\begin{gather*}
 \sup_{j \in \mathbb{N}} \big\|(1-\Delta)^{-M}M_{\psi_j}\big\|_{q,\infty}\leq \sup_{j \in \mathbb{N}} \big( \big\| M_{\psi_j} \big\|_\infty \! \cdot \big\|\big(1-\Delta_{g^{x_j}}\big)^{-M}\big\|_{q,\infty} \big)\\
 \hphantom{\sup_{j \in \mathbb{N}} \big\|(1-\Delta)^{-M}M_{\psi_j}\big\|_{q,\infty}\leq}{}
 + \sup_{j \in \mathbb{N}} \big(\big\| (1-\Delta)^{-M}\!\big[(1-\Delta)^M, M_{\psi_j}\big] \big\|_\infty \! \cdot \big\|\big(1-\Delta_{g^{x_j}}\big)^{-M}\big\|_{q,\infty} \big) \\
 \hphantom{\sup_{j \in \mathbb{N}} \big\|(1-\Delta)^{-M}M_{\psi_j}\big\|_{q,\infty}}{}
 <\infty.
\tag*{\qed}
\end{gather*}
\renewcommand{\qed}{}
\end{proof}

It follows from this proposition that for every $q>0$ and every $j,$ there exists $M>0$ and a~constant $C_M$ independent of $j$ such that
\begin{equation}\label{weak_bounded_cwikel_estimate}
 \big\|M_{f\psi_j}(1-\Delta)^{-M}\big\|_{q,\infty} \leq C_M\|f\|_{L_{\infty}(B(x_j,r_0))}.
\end{equation}

Let $\big\{\psi_{j}^{(0)}\big\}_{j=0}^\infty, \big\{\psi_{j}^{(2)}\big\}_{j=0}^\infty , \ldots , \big\{\psi_{j}^{(N)}\big\}_{j=0}^\infty$ be a partition of $\{\psi_{j}\}_{j=0}^\infty$ into disjointly supported subfamilies. That is, for all $0\leq k\leq N,$ the functions $\big\{\psi_{j}^{(k)}\big\}_{j=0}^\infty$ are disjointly supported, and for every $j\geq 0$ there exists
a unique $0\leq k\leq N$ such that $\psi_{j} \in \big\{\psi_l^{(k)}\big\}_{l=0}^\infty.$

\begin{Theorem}%\label{T: Cwikel estimate}
 Let $0<p<2.$ For sufficiently large $M$, we have
 \begin{align*}
 &\big\|M_f(1-\Delta)^{-M}\big\|_{p,\infty}\leq C_{p,M,N}\|f\|_{\ell_{p,\infty}(L_\infty)},\\
 &\big\|M_f(1-\Delta)^{-M}\big\|_{p}\leq C_{p,M,N}'\|f\|_{\ell_{p}(L_\infty)}.
 \end{align*}
\end{Theorem}
\begin{proof}
 We prove the first inequality, the second can be proved analogously. Let $f \in \ell_{p,\infty}(L_\infty).$
 Since $\{\psi_j\}_{j=1}^\infty$ is a partition of unity, we have
 \[
 f = \sum_{j=0}^\infty \psi_jf = \sum_{k=0}^N \sum_{j=0}^\infty \psi_{j}^{(k)}f.
 \]
 By the quasi-triangle inequality, there exists $C_{N,p}$ such that
 \[
 \big\|M_f(1-\Delta)^{-M}\big\|_{p,\infty} \leq C_{N,p}\sum_{k=0}^N \Bigg\| \sum_{j=0}^\infty M_{\psi_j^{(k)}f} (1-\Delta)^{-M}\Bigg\|_{p,\infty}.
 \]
 The operators $\big\{M_{\psi_j^{(k)}f}(1-\Delta)^{-M}\big\}_{j=0}^\infty$ are left-disjoint. Hence, Corollary~\ref{disjointification_corollary} implies that
 there exists $q>0$ such that
 \[
 \big\|M_f(1-\Delta)^{-M}\big\|_{p,\infty} \leq C_{N,p,q}\sum_{k=0}^N \big\| \big\{\big\|M_{f\psi_j^{(k)}}(1-\Delta)^{-M}\big\|_{q,\infty}\big\}_{j=0}^\infty\big\|_{\ell_{p,\infty}}.
 \]
 From~\eqref{weak_bounded_cwikel_estimate}, it follows that if $M$ is sufficiently large (depending on $q$), we have
 \[
 \big\|M_f(1-\Delta)^{-M}\big\|_{p,\infty}\leq C_{N,p,q}\big\| \big\{\|f\|_{L_{\infty}(B(x_j,r_0))}\big\}_{j=0}^\infty\big\|_{\ell_{p,\infty}}.
 \]
 The latter is the definition of the $\ell_{p,\infty}(L_\infty)$ quasinorm.
\end{proof}

\begin{Corollary} \label{C:Cwikel} Let $0<p<2$. Let $P$ be a self-adjoint, lower-bounded $P \in EBD^m(X)$. We have $\exp(-tP)M_f \in \mathcal{L}_{p,\infty}$ whenever $f\in \ell_{p,\infty}(L_\infty)$, and $\exp(-tP)M_f \in \mathcal{L}_{p}$ whenever $f\in \ell_{p}(L_\infty)$.
\end{Corollary}
\begin{proof}
By the preceding theorem, we already have that for $f\in \ell_{p,\infty}(L_\infty)$ and sufficiently large~$M$, $M_f(1-\Delta)^{-M} \in \mathcal{L}_{p,\infty}$. Noting that $P+C>0$ and is therefore an invertible operator on~$L_2(X)$ for some $C\in \mathbb{R}$, by~\cite[Proposition~4.4]{Kordyukov1991} it follows that $(P+C)^{-1}$ maps~$H_s(X)$ boundedly into $H_{s+m}(X)$. By~\cite[Theorem~3.9]{Kordyukov1991}, $(1-\Delta)^M$ maps $H_s(X)$ boundedly into~$H_{s-2M}(X)$. Therefore, we can find $N\in \mathbb{N}$ large enough such that $(P+C)^{-N}(1-\Delta)^M$ is a~bounded operator on $L_2(X)$. Noting that $\exp(-tP) (P+C)^N$ is a bounded operator for any $N$ by the functional calculus on unbounded operators, the claim follows. The case for $f \in \ell_p(L_\infty)$ is proven analogously.
\end{proof}

This corollary will eventually make it possible to apply Theorem~\ref{main_theorem} on $P$ and $W = M_w$ for some $w \in \ell_{1,\infty}(L_\infty)$. Finding such a $w$ is the first step. The second step is to show that also $\exp(-tP)[P,M_w] \in \mathcal{L}_1$, but this will require geometric conditions on the manifold $X$.

\begin{Lemma}\label{Grimaldi}
Let $(X,g)$ be a complete connected Riemannian manifold of bounded geometry. Then for any fixed $r\in \mathbb{R}$,
\[
\liminf_{R \to \infty} \frac{ \abs{B(x_0, R-r)}}{ \abs{B(x_0, R+r)}} > 0.
\]
\begin{proof}
The paper~\cite{GrimaldiPansu2011} proves that for manifolds of bounded geometry, we have for fixed $R\geq 1$,
\[
\frac{1}{L} \leq \abs{B(x_0, R+2)}-\abs{B(x_0, R+1)} \leq L(\abs{B(x_0, R+1)} - \abs{B(x_0, R))}
\]
for some constant $L>0$ independent of $R$.

Now with this inequality, one can show by induction that
\[
\frac{\abs{B(x_0,R+k)}}{\abs{B(x_0,R)}} \leq (1+L)^k, \qquad R \geq 1.
\] For $k=1$,
\begin{align*}
& \frac{\abs{B(x_0,R+1)}}{\abs{B(x_0,R)}}\leq 1 + L \frac{\abs{B(x_0,R)} - \abs{B(x_0,R-1)}}{\abs{B(x_0,R)}}\leq 1 + L.
\end{align*}
Suppose that $\frac{\abs{B(x_0,R+k)}}{\abs{B(x_0,R)}} \leq (1+L)^k$ for some $k$, then \begin{align*}
 \frac{\abs{B(x_0,R+k+1)}}{\abs{B(x_0,R)}} &\leq \frac{\abs{B(x_0,R+k)}}{\abs{B(x_0,R)}} + L \frac{\abs{B(x_0,R+k)} - \abs{B(x_0,R+k-1)}}{\abs{B(x_0,R)}}\\
 &\leq (1+L)^k + L (1+L)^k = (1+L)^{k+1}.
\end{align*}

Therefore,
\begin{align*}
 &\liminf_{R \to \infty} \frac{ \abs{B(x_0, R-r)}}{ \abs{B(x_0, R+r)}} = \liminf_{R \to \infty} \frac{ \abs{B(x_0, R)}}{ \abs{B(x_0, R+2r)}} > \frac{1}{(1+L)^K},
\end{align*}
where $K$ is some integer greater than $2r$.
\end{proof}
\end{Lemma}

\begin{Lemma}\label{wl1infty}
Let $(X,g)$ be a complete connected Riemannian manifold of bounded geometry. Then the function
\[
 w(x) = (1+\abs{B(x_0,d_X(x,x_0))})^{-1},\qquad x \in X,
\]
is an element of $\ell_{1,\infty}(L_\infty)(X).$
\begin{proof}
We denote $d_X(x,x_0)$ by $r(x)$. Note that for any $x \in B(x_k, r_0)$, we have $r(x) \geq \abs{x_k}-r_0$ by the triangle inequality, and hence
\[
\|w\|_{L_\infty(B(x_k,r_0))} \leq (1+\abs{B(x_0,r(x_k)-r_0)})^{-1}.
\]
Without loss of generality, order the $x_j$ such that $r(x_1) \leq r(x_2) \leq \cdots$. We only need to show that
\[
 \frac{1}{1+|B(x_0,r(x_k)|} = O\big(k^{-1}\big)
\]
or equivalently that
\[
 \inf_{k\geq 1} k|B(x_0,r(x_k))| > 0.
\]

Note that as in Remark~\ref{R: metric balls lower bound}, we have that $\inf_{j \in \mathbb{N}} \abs{B(x_j, r_0)} > 0$ and hence
\[
 \abs{B(x_0,\abs{x_k}-r_0)} = \inf_{j \in \mathbb{N}} \abs{B(x_j, r_0)} \cdot \frac{ \abs{B(x_0, \abs{x_k}-r_0)}}{ \abs{B(x_0, \abs{x_k}+r_0)}} \cdot \frac{\abs{B(x_0, \abs{x_k}+r_0)}}{\inf_{j \in \mathbb{N}} \abs{B(x_j, r_0)}}.
\]
By the ordering of the $x_j$, we know that all the balls $B(x_j, r_0)$ from $j=1$ up to and including~${j=k}$ are contained in $B(x_0, \abs{x_k}+r_0)$.
Hence,
\begin{align*}
& k \cdot \inf_{m\in \mathbb{N}} \abs{B(x_m, r_0)} \leq \sum_{j=1}^k \abs{B(x_j, k)} \leq (N+1) \abs{B(x_0, \abs{x_k}+r_0)},
\end{align*}
since any ball can only intersect at most $N$ other balls. We thus get
\[
\abs{B(x_0, \abs{x_k}+r_0)} \geq \frac{k}{N+1}\inf_{j \in \mathbb{N}} \abs{B(x_j, r_0)}.
\]
We will now show that for $k$ large enough, $\frac{ \abs{B(x_0, \abs{x_k}-r_0)}}{ \abs{B(x_0, \abs{x_k}+r_0)}}$ is bounded below away from zero. By Lemma~\ref{Grimaldi}, we have
\[
\liminf_{R \to \infty} \frac{ \abs{B(x_0, R-r_0)}}{ \abs{B(x_0, R+r_0)}} > 0,
\]
and so there must be some $R_0$ such that for $R \geq R_0$, we have $\frac{ \abs{B(x_0, R-r_0)}}{ \abs{B(x_0, R+r_0)}} > \delta > 0$. We claim that we can take $K$ large enough such that $\abs{x_k} \geq R_0$ for $k \geq K$. Indeed, by analogous reasoning as before, at most $K:= (N+1)\frac{\abs{B(x_0, R_0+r_0)}}{\inf_j \abs{ B(x_j, r_0)}}$ points $x_j$ can be inside the ball $\abs{B(x_0, R_0)}$, thus~$\abs{x_k} > R_0$ for $k > K$.

Gathering the results above, we get the existence of some constant $C$ such that for $k \geq K$ we have
\[
\abs{B(x_0,\abs{x_k}-r_0)} \geq C k.
\]
This means that
\[
\|w\|_{L_\infty(B(x_k,r_0))} \leq (1+C k)^{-1},
\]
proving that $w \in \ell_{1,\infty}(L_\infty)(X).$
\end{proof}
\end{Lemma}

We will now mold Property~(D) into the form that we will apply it in.
\begin{Lemma}\label{L:NewPropD}
Let $(X,g)$ be a complete connected Riemannian manifold of bounded geometry. If $X$ has Property~$(D)$, that is if
\begin{equation}\label{eq3:Condition1}
\left \{ \frac{\abs{\partial B(x_0, k)}}{\abs{B(x_0,k)}}\right\}_{k\in \mathbb{N}} \in \ell_2(\mathbb{N})
\end{equation}
and
\begin{equation}\label{eq3:Condition2}
 \lim_{R\to \infty} \frac{\frac{\partial}{\partial r}\big |_{r=R} |\partial B(x_0, r)|}{|\partial B(x_0,R)|} = 0,
\end{equation}
then also for every $h > 0$,
\begin{equation}
 \label{eq:Condition1Mod}
\left \{ \frac{\sup_{s \in [0, h]}\abs{\partial B(x_0, k+s)}}{\abs{B(x_0,k)}}\right\}_{k\in \mathbb{N}} \in \ell_2(\mathbb{N}).
\end{equation}
\end{Lemma}
\begin{proof}
 First, let $\varepsilon >0$ and choose $R$ large enough so that $r \geq R$ implies that
 \[
 \frac{\partial}{\partial s}\bigg|_{s=r} |\partial B(x_0, s)| \leq \varepsilon|\partial B(x_0, r)|,
 \]
 which is possible due to equation~\eqref{eq3:Condition2}. Then, for $r \geq R$, $\delta > 0$,
 \begin{align*}
 \frac{|\partial B(x_0, r + \delta)|}{|B(x_0, r)|} &= \frac{|\partial B(x_0, r)|}{|B(x_0, r)|} + \frac{\int_r^{r+\delta} \frac{\partial}{\partial s} |\partial B(x_0, s)| \,{\rm d}s}{|B(x_0,r)|}\\
 &\leq \frac{|\partial B(x_0, r)|}{|B(x_0, r)|} + \varepsilon \delta \sup_{\gamma \in [0,\delta]}\frac{ |\partial B(x_0, r+\gamma)| }{|B(x_0,r)|}.
 \end{align*}
 Taking the supremum over $\delta \in [0, h]$ on both sides and rearranging gives
 \[
 (1-\varepsilon h) \sup_{s \in [0,\delta]}\frac{ |\partial B(x_0, r+s)| }{|B(x_0,r)|} \leq \frac{|\partial B(x_0, r)|}{|B(x_0, r)|}.
 \]
 This shows that~\eqref{eq3:Condition1} and~\eqref{eq3:Condition2} together imply~\eqref{eq:Condition1Mod}
\end{proof}

\begin{Lemma}\label{L:Property_D}
 Let $(X,g)$ be a non-compact Riemannian manifold with bounded geometry and Property~$($D$)$. Then \begin{equation}\label{Condition1}
 \sum_{k=0}^\infty \frac{\sup_{s \in [-1, 2]}\abs{\partial B(x_0, (k+s)r_0)}^2}{(1+\abs{B(x_0,(k-1)r_0)})^2} <\infty.
\end{equation}
\end{Lemma}
\begin{proof}
 Since $X$ has Property~(D), Lemma~\ref{L:NewPropD} gives that
 \[
 \sum_{k=1}^\infty \frac{\sup_{s \in [0, 3]}\abs{\partial B(x_0, k+s)}^2}{\abs{B(x_0,k)}^2} < \infty.
 \]
 Recall that we can assume $r_0 \leq 1$. Then
 \begin{align*}
 &\sum_{k=0}^\infty \frac{\sup_{s \in [-1, 2]}\abs{\partial B(x_0, (k+s)r_0)}^2}{(1+\abs{B(x_0,(k-1)r_0)})^2} \\
 &\qquad\quad{} = \sum_{N=0}^\infty \sum_{N \leq kr_0 \leq N+1} \frac{\sup_{s \in [-1, 2]}\abs{\partial B(x_0, (k+s)r_0)}^2}{(1+\abs{B(x_0,(k-1)r_0)})^2}\\
 &\qquad\quad{} \leq \sum_{N=0}^\infty \sum_{N \leq kr_0 \leq N+1} \frac{\sup_{s \in [-1, 2]}\abs{\partial B(x_0, N+s)}^2}{(1+\abs{B(x_0,N-1)})^2}\\
 &\qquad\quad{} \leq \left\lceil \frac{1}{r_0}\right\rceil \left(2\cdot \sup_{s \in [0, 3]}\abs{\partial B(x_0, s)}^2+ \sum_{k=1}^\infty \frac{\sup_{s \in [0, 3]}\abs{\partial B(x_0, k+s)}^2}{\abs{B(x_0,k)}^2}\right) <\infty.
\tag*{\qed}
\end{align*}
\renewcommand{\qed}{}
\end{proof}

Note that in the next lemma, we do not need uniform ellipticity of the differential operators considered, although we will consider operators $L \in BD^2(X)$ which lack a constant term. The meaning of this condition
is that in a system of normal coordinates $(U,\phi),$ where $U = B(x,r_0)$, we have
\[
 L = \sum_{0<|\alpha|\leq 2} a_{\alpha,x}(y)D^{\alpha}.
\]
Equivalently, $L1 = 0.$ The important feature of these operators is that if $P \in BD^2(X)$ and $f$ is a smooth function with uniformly bounded derivatives, then $[P,M_f] = [L, M_f]$ for a differential operator $L\in BD^2(X)$ with no constant term.

\begin{Lemma}\label{L:Dwl1}
Let $(X,g)$ be a non-compact Riemannian manifold with bounded geometry satisfying Property~$($D$)$ $($Definition~$\ref{D:Property_D})$. Then $Lw \in \ell_1(L_\infty)(X)$ for any $L \in BD^2(X)$ that lacks a~constant term.
\end{Lemma}
\begin{proof}
For any $x \in X$, we can take a neighbourhood of normal coordinates $(U, \phi)$ in which $L$ takes the following form by definition:
\[
\sum_{|\alpha|= 1,2} a_{\alpha, x}(y) D^\alpha,
\]
where for any multi-index $\beta$
\[
\big\lvert\partial^\beta a_{\alpha, x}(0)\big\rvert \leq C_{\alpha, \beta}, \qquad x \in X.
\]
Now denote $r(x) := d_X(x_0,x)$, $\tilde{w}(r) = (1+|B(x_0,r)|)^{-1}$ so that $w = \tilde{w} \circ r$. Note that $\tilde{w}$ is a smooth function on $[0,\infty),$ and the combination of~\cite[Lemma~III.4.4]{Sakai1996} and~\cite[Proposition~III.4.8]{Sakai1996} gives that $r$ is smooth almost everywhere with $\norm{\nabla r} \leq 1$ almost everywhere. We therefore have almost everywhere
\begin{align*}
&\abs{Lw(x)}
 \leq \bigg | \sum_{|\alpha|= 1} a_{\alpha, x}(0) \big(\partial^\alpha \big(\tilde{w}\circ r \circ \phi^{-1}\big)\big)(0) \bigg |
 + \bigg | \sum_{\substack{|\alpha|=2\\\alpha = \beta + \gamma\\|\beta|=|\gamma|= 1}} a_{\alpha, x}(0) \big(\partial^\beta \partial^\gamma \big(\tilde{w}\circ r \circ \phi^{-1}\big)\big)(0)\bigg|\\
& \qquad{}= \bigg | \sum_{|\alpha|= 1} a_{\alpha, x}(0) \big(\partial^\alpha \big(r \circ \phi^{-1}\big)\big)(0) \tilde{w}'(r(x)) \bigg |\\
&\qquad\quad{} + \bigg | \sum_{\substack{|\alpha|=2\\\alpha = \beta + \gamma\\|\beta|=|\gamma|= 1}} a_{\alpha, x}(0) \big( \tilde{w}''(r(x)) \cdot \partial^\beta \big(r \circ \phi^{-1}\big)(0) \cdot \partial^\gamma \big(r \circ \phi^{-1}\big)(0) \\
&\qquad \quad{} \hphantom{\bigg | \sum_{\substack{|\alpha|=2\\\alpha = \beta + \gamma\\|\beta|=|\gamma|= 1}}}{}
 + \tilde{w}'(r(x)) \cdot \partial^\beta \partial^\gamma \big(r \circ \phi^{-1}\big)(0)\big) \bigg |\\
&\qquad{}\leq \bigg | \sum_{|\alpha|= 1} a_{\alpha, x}(0) \bigg | \cdot \big|\tilde{w}'(r(x))\big| + \bigg |\sum_{\substack{|\alpha|=2\\\alpha = \beta + \gamma\\|\beta|=|\gamma|= 1}} a_{\alpha,x}(0) \bigg|\cdot \big|\tilde{w}''(r(x))\big| \\
&\qquad\quad{}+ C\bigg |\sum_{\substack{|\alpha|=2\\\alpha = \beta + \gamma\\|\beta|=|\gamma|= 1}} a_{\alpha,x}(0) \bigg|\cdot \big|\tilde{w}'(r(x))\big| \\
&\qquad {}\leq (1+C)\cdot\big( \big|\tilde{w}'(r(x))\big|+\big|\tilde{w}''(r(x))\big|\big),
\end{align*}
where we have used the chain rule and
\[
\big|\big(\partial^\alpha \bigl(r \circ \phi^{-1}\bigr)\big)(0)\big| = \left|\frac{\partial}{\partial x^\alpha}\bigg|_x (r)\right| \leq \norm{\nabla r(x)} \leq 1
\]
\big(because $(\nabla r)_\alpha = \sum_{\beta}g_{\alpha \beta} \partial^\beta r$ and in normal coordinates at zero, $g_{\alpha,\beta}(0) = \delta_{\alpha,\beta}$\big)
in addition to
\[
\big|\partial^\beta \partial^\gamma \big(r \circ \phi^{-1}\big)(0)\big| \leq \| \mathrm{Hess}\ r \| \leq C
\]
for some constant $C$, since the Hessian of $r$ is uniformly bounded~\cite[Theorem~6.5.27]{Petersen2006}.

Since $w$ and $\tilde{w}$ are smooth, we have everywhere
\[
\abs{Lw(x)} \leq (1+C)\cdot\big(\big|\tilde{w}'(r(x))\big| + \big|\tilde{w}''(r(x))\big|\big).
\]
 Therefore,
\begin{align*}
 &\norm{Lw}_{L_\infty(B(x_j, r_0))} \leq (1+C)\cdot \Big(\sup_{s \in [-r_0, r_0]} \abs{\tilde{w}'(r(x_j)+s)} + \abs{\tilde{w}''(r(x_j)+s)} \Big)\\
 &\qquad{}\leq (1+C)\cdot \Bigg(\sup_{s \in [-r_0, r_0]} \frac{\abs{\partial B(x_0, r(x_j)+s)}}{(1+\abs{B(x_0,r(x_j)+s)})^2} + \sup_{s \in [-r_0, r_0]} \frac{\frac{{\rm d}}{{\rm d}R}\Big|_{R= r(x_j)+s}\abs{\partial B(x_0, R)}}{(1+\abs{B(x_0,r(x_j)+s)})^2} \\
 &\qquad\quad{} + 2 \sup_{s \in [-r_0, r_0]} \frac{\abs{\partial B(x_0, r(x_j)+s)}^2}{(1+\abs{B(x_0,r(x_j)+s)})^3}\Bigg).
\end{align*}
Next, we calculate
\begin{align*}
 &\sum_{j\in \mathbb{N}} \sup_{s \in [-r_0, r_0]} \frac{\abs{\partial B(x_0, r(x_j)+s)}}{(1+\abs{B(x_0,r(x_j)+s)})^2}\\
 &\qquad {}= \sum_{k=0}^\infty \sum_{\{j\colon kr_0 \leq r(x_j)< (k+1)r_0\}} \sup_{s \in [-r_0, r_0]} \frac{\abs{\partial B(x_0, r(x_j)+s)}}{(1+\abs{B(x_0,r(x_j)+s)})^2}\\
 &\qquad {}\leq \sum_{k=0}^\infty \abs{\{j\colon k r_0 \leq r(x_j)< (k+1)r_0\}}\cdot \sup_{l \in [-1, 2]} \frac{\abs{\partial B(x_0, (k+l)r_0}}{(1+\abs{B(x_0,(k+l)r_0)})^2}.
\end{align*}
With a similar calculation as in the proof of Lemma~\ref{wl1infty}, we have
\begin{align*}
 &\inf_{m \in \mathbb{N}} \abs{B(x_m, r_0)} \cdot \abs{\{j: k r_0 \leq r(x_j)< (k+1)r_0\}} \\
 &\qquad {}= \sum_{\{j : k r_0 \leq r(x_j)< (k+1)r_0\}} \inf_{m \in \mathbb{N}}\abs{B(x_m, r_0)}\\
 &\qquad {}\leq \sum_{\{j : k r_0 \leq r(x_j)< (k+1)r_0\}} \abs{B(x_j, r_0)}\\
 &\qquad {}\leq (N+1) (\abs{B(x_0, (k+2)r_0)} - \abs{B(x_0, (k-1)r_0)})\\
 &\qquad {}\leq 3(N+1)r_0 \sup_{l\in [-1,2]} \abs{\partial B(x_0, (k+l)r_0)},
\end{align*}
since all the balls $B(x_j, r_0)$ with $kr_0 \leq r(x_j)\leq (k+1)r_0$ are contained in the annulus $B(x_0, (k+2)r_0)\setminus B(x_0, (k-1)r_0)$, and balls can intersect at most $N$ other balls.

Using Lemma~\ref{L:Property_D}, we can infer that the expression~\eqref{Condition1} is finite, and so
\begin{align*}
 \sum_{j\in \mathbb{N}} \sup_{s \in [-r_0, r_0]} \frac{\abs{\partial B(x_0, r(x_j)+s)}}{(1+\abs{B(x_0,r(x_j)+s)})^2} &\leq C' \sum_{k=0}^\infty \frac{\sup_{l \in [-1, 2]}\abs{\partial B(x_0, (k+l)r_0}^2}{(1+\abs{B(x_0,(k-1)r_0)})^2} < \infty.
\end{align*}

With analogous calculations, we also have that
\[
\sum_{j\in \mathbb{N}}\sup_{s \in [-r_0, r_0]} \frac{\frac{{\rm d}}{{\rm d}R}\Big|_{R= r(x_j)+s}\abs{\partial B(x_0, R)}}{(1+\abs{B(x_0,r(x_j)+s)})^2} <\infty
\]
and
\[
\sum_{j \in \mathbb{N}}\sup_{s \in [-r_0, r_0]} \frac{\abs{\partial B(x_0, r(x_j)+s)}^2}{(1+\abs{B(x_0,r(x_j)+s)})^3} < \infty.
\]
We conclude that
\[
\|Lw\|_{\ell_1(L_\infty)} = \sum_{j\in \mathbb{N}} \| Lw \|_{L_{\infty}(B(x_j,r_0))} < \infty.
\tag*{\qed}
\]
\renewcommand{\qed}{}
\end{proof}

\begin{Corollary}\label{C:CwikelL1}
Let $P\in EBD^2(X)$ be a self-adjoint lower-bounded operator. Let $X$ be a~non-compact Riemannian manifold with Property~$($D$)$. Then $\exp(-tP)[P, M_w] \in \mathcal{L}_1$.
\end{Corollary}
\begin{proof}
Take $P \in EBD^2(X).$ Then, using the expression $P = \sum_{|\alpha|\leq 2} a_{\alpha,x_j}(y)\partial_y^{\alpha}$ in the coordinate chart $B_j$, we have
\begin{align*}
&\big\|(1-\Delta)^{-M-1}[P,M_w] M_{\psi_j}\big\|_1 \\
&\qquad{}
  \leq \bigg\|(1-\Delta)^{-M-1}\sum_{|\alpha|\leq 2} \big[M_{a_{\alpha,x_j}} \partial^\alpha,M_w\big] M_{\psi_j}\bigg\|_1\\
&\qquad {}\leq \big\|(1-\Delta)^{-M-1} M_{\psi_j} M_{D w}\big\|_1\\
&\qquad\quad{} + \bigg\|(1-\Delta)^{-M-1}\sum_{\substack{|\alpha|=2\\\alpha = \beta + \gamma\\|\beta|=|\gamma|= 1}} M_{a_{\alpha,x_j}} \big(\partial^\gamma M_{\partial^\beta w}+\partial^\beta M_{\partial^\gamma w}\big) M_{\psi_j}\bigg\|_1\\
&\qquad {}\leq \big\|(1-\Delta)^{-M-1} M_{\psi_j} M_{D w}\big\|_1\\
&\qquad\quad{} + \bigg\|(1-\Delta)^{-M-1}\sum_{|\beta|=1}\sum_{|\gamma|=1} M_{a_{\beta+\gamma,x_j}} \partial^\gamma M_{\partial^\beta w} M_{\psi_j}\bigg\|_1\\
&\qquad\quad{} + \bigg\|(1-\Delta)^{-M-1}\sum_{|\beta|=1} M_{a_{2\beta,x_j}} \partial^\beta M_{\partial^\beta w} M_{\psi_j}\bigg\|_1\\
&\qquad {}=: I + II + III,
\end{align*}
where we have denoted $D$ for the differential operator given near $x_j$ by
\[
 D = \sum_{0<|\alpha|\leq 2} a_{\alpha,x_j}(y)\partial_y^{\alpha}.
\]
That is, $D$ is equal to $P$ without constant terms. Since $\psi_j$ is by definition supported in $B_j$, we have
\begin{align*}
I &= \big\|(1-\Delta)^{-M-1} M_{\psi_j} M_{D w}\big\|_1\\
&= \big\|(1-\Delta)^{-M-1} M_{\psi_j} M_{\chi_{B_j}}M_{D w}\big\|_1\\
&\leq \big\| (1-\Delta)^{-1}\big\|_\infty \big\|(1-\Delta)^{-M} M_{\psi_j}\big\|_{1} \| M_{\chi_{B_j}}M_{D w} \|_\infty\\
&=\big\| (1-\Delta)^{-1}\big\|_\infty \big\|(1-\Delta)^{-M} M_{\psi_j}\big\|_{1} \| Dw \|_{L_\infty (B_j)}.
\end{align*}
By Lemma~\ref{L:Dwl1}, $\big\{\| Dw \|_{L_\infty (B_j)} \big\}_{j\in \mathbb{N}}\in \ell_1$,
\begin{align*}
II &= \bigg\|(1-\Delta)^{-M-1}\sum_{|\beta|=1}\sum_{|\gamma|=1} M_{a_{\beta+\gamma,x}} \partial^\gamma M_{\partial^\beta w} M_{\psi_j}\bigg\|_1\\
&\leq \bigg\|\sum_{|\beta|=1}\sum_{|\gamma|=1} (1-\Delta)^{-M-1} [M_{a_{\beta+\gamma,x}}, \partial^\gamma] M_{\chi_{B_j}} M_{\partial^\beta w} M_{\psi_j}\bigg\|_1\\
& \quad + \bigg\|\sum_{|\beta|=1}\sum_{|\gamma|=1} (1-\Delta)^{-M-1} \partial^\gamma M_{\chi_{B_j}}M_{a_{\beta+\gamma,x}} M_{\partial^\beta w} M_{\psi_j}\bigg\|_1\\
&= \bigg\|\sum_{|\beta|=1}\sum_{|\gamma|=1} (1-\Delta)^{-M-1} M_{\partial^\gamma a_{\beta+\gamma,x}} M_{\chi_{B_j}} M_{\partial^\beta w} M_{\psi_j}\bigg\|_1\\
& \quad + \bigg\|\sum_{|\beta|=1}\sum_{|\gamma|=1} (1-\Delta)^{-M-1} \partial^\gamma M_{\chi_{B_j}}M_{a_{\beta+\gamma,x}} M_{\partial^\beta w} M_{\psi_j}\bigg\|_1\\
&\leq \sum_{|\beta|=1}\sum_{|\gamma|=1} \big\| (1-\Delta)^{-1} \partial^\gamma M_{\chi_{B_j}} \big\|_{\infty} \big\| (1-\Delta)^{-M}M_{\psi_j} \big\|_1 \| a_{\beta+\gamma,x}\|_{L_\infty (B_j)} \big\| \partial^\beta w\big\|_{L_\infty(B_j)}\\
& \quad + \sum_{|\beta|=1}\sum_{|\gamma|=1} \big\|(1-\Delta)^{-1}\big\|_\infty \big\| (1-\Delta)^{-M} M_{\psi_j}\big\|_{1} \big\|\partial^\gamma a_{\beta+\gamma,x}\big\|_{L_\infty (B_j)} \big\| \partial^\beta w \big\|_{L_\infty (B_j)}.
\end{align*}
Suppose $y$ is in the normal neighbourhood $(U_x,\phi_x)$ of $x$. Then, $D$ being a local operator, at $y$ (taking $\tilde{y} = \phi_x(y)$), there are two different expressions of $D$ in normal coordinates:
\[
D = \sum_{|\alpha| \leq 2} a_{\alpha, x}(\tilde{y}) \partial_x^\alpha = \sum_{|\beta|\leq 2} a_{\beta, y}(0) \partial_y^\beta.
\]
Hence, we can express $a_{\alpha, x}(\tilde{y})$ in terms of $a_{\beta, y}(0)$ and transition functions $\frac{\partial^\alpha y}{\partial x^\alpha}$. These are uniformly bounded by the definition of a differential operator with bounded coefficients and \cite[Proposition~1.3]{Kordyukov1991}. Therefore, $\|a_{\beta+\gamma,x}\|_{L_\infty (B_j)} $ and also $\|\partial^\gamma a_{\beta+\gamma,x}\|_{L_\infty (B_j)} $ are uniformly bounded in $j$. Likewise (taking $|\beta| = 1$),
\[
\big|\partial^\beta_x w (y)\big| = \bigg|\sum_{|\alpha|=1} \frac{\partial^\alpha y}{\partial x^\alpha} \partial^\alpha_y w(y) \bigg|
\leq C_d \| \nabla r\| \big|\tilde{w}'(y)\big|,
\]
and hence $\big\{\big\| \partial^\beta w \big\|_{L_\infty(B_j)}\big\}_{j \in \mathbb{N}} \in \ell_1$ by the arguments in the proof of Lemma~\ref{L:Dwl1}.

Finally,
\begin{align*}
&\big\| (1-\Delta)^{-1} \partial^\gamma M_{\chi_{B_j}} \big\|_{\infty}\leq \sup_{ |\alpha|=1} \bigg\|\frac{\partial^\alpha y}{\partial x^\alpha}\bigg\|_{L_{\infty}(B_j)} \big\| (1-\Delta)^{-1} |\nabla| \big\|_\infty < \infty.
\end{align*}
The same estimates hold for $III$.

Combining everything, we have
\begin{align*}
\big\|(1-\Delta)^{-M-1}[P,M_w]\big\|_1 &\leq \sum_{j \in \mathbb{N}} \big\|(1-\Delta)^{-M-1}[D,M_w] M_{\psi_j}\big\|_1 < \infty.
\end{align*}
By the arguments in the proof of Corollary~\ref{C:Cwikel}, this implies that also $\exp(-tP)[P,M_w] \in \mathcal{L}_1$.
\end{proof}

Gathering all results in this section, let $X$ be a non-compact manifold of bounded geometry with Property~(D). Let $P\in EBD^2(X)$ be self-adjoint and lower-bounded. Then $\exp(-tP)M_w \in \mathcal{L}_{1,\infty}$ by the Cwikel estimate in Corollary~\ref{C:Cwikel} and Lemma~\ref{wl1infty}. Corollary~\ref{C:CwikelL1} states that $\exp(-tP)[P,M_w] \in \mathcal{L}_1$. Theorem~\ref{main_theorem} then gives that
 \[
 \mathrm{Tr}_{\omega}\big({\rm e}^{-tP}M_w\big) = \lim_{\varepsilon\to 0} \varepsilon \mathrm{Tr}\big({\rm e}^{-tP}\chi_{[\varepsilon,\infty)}(M_w)\big)
 \]
if the limit on the right-hand side exists. If we assume that $P$ admits a density of states, we do in fact get the existence of the limit
\begin{align*}
&\int_{\mathbb{R}} {\rm e}^{-t\lambda}\, {\rm d}\nu_P(\lambda)=\lim_{R\to\infty} \frac{1}{|B(x_0,R)|}\mathrm{Tr}\big({\rm e}^{-tP}M_{\chi_{B(x_0,R)}}\big)
 = \lim_{\varepsilon\to 0} \varepsilon \mathrm{Tr}\big({\rm e}^{-tP}\chi_{[\varepsilon,\infty)}(M_w)\big).
\end{align*}

Note that the above calculation assumes that the volume of $X$ is infinite, which is equivalent with $X$ being non-compact (see Remark~\ref{R: metric balls lower bound}).

Hence,
\[
\mathrm{Tr}_{\omega}\big({\rm e}^{-tP}M_w\big) = \int_{\mathbb{R}} {\rm e}^{-t\lambda}\, {\rm d}\nu_P(\lambda),\qquad t>0.
\]
From this, we can easily deduce by a density argument that
\[
\mathrm{Tr}_{\omega}(f(P)M_w) = \int_{\mathbb{R}} f(\lambda)\, {\rm d}\nu_P(\lambda), \qquad f \in C_c(\mathbb{R}).
\]
For details on the required density argument, see~\cite[Remark 6.3]{AMSZ}.
This concludes the proof of Theorem~\ref{T: main manifold thm}.

\section{Roe's index theorem}\label{S:Roe}

The results of the preceding section were stated for operators $P$ acting on scalar valued functions on $X.$ Identical results, with the same proofs, apply to operators acting between sections of vector bundles of bounded geometry. In the terminology of Shubin~\cite{Shubin1992}, a rank $N$ vector bundle $\pi\colon S\to X$ is said to have bounded geometry if in every coordinate chart $\{B(x_j,r_0)\}_{j=0}^\infty$ (as defined in Section~\ref{S: Manifolds}) $E$ has a trivialisation \[
 \pi^{-1}(B(x_j,r_0)) \approx B(x_j,r_0)\times \mathbb{R}^N
\]
such that the transition functions
\[
 t_{j,k}\colon \ B(x_j,r_0)\times \mathbb{R}^N\cap B(x_k,r_0)\times \mathbb{R}^N \to B(x_j,r_0)\times \mathbb{R}^N\cap B(x_k,r_0)\times \mathbb{R}^N
\]
have uniformly bounded derivatives in the exponential normal coordinates around $x_j$ or $x_k.$ See also Eichhorn~\cite[p.~65]{Eichhorn2008}.

For our purposes, we will assume that $S\to X$ is equipped with a Hermitian metric $h,$ which is assumed to be a $C^\infty$-bounded section of the bundle $\overline{S}\otimes S$ in the terminology of Shubin~\mbox{\cite[Appendix 1]{Shubin1992}}. We denote $L_2(X,S)$ for the Hilbert space of square integrable sections of $S$ with respect to the volume form of $X$ and the Hermitian metric $h.$

Shubin defines elliptic differential operators
acting on sections of vector bundles of bounded geometry.
Given a vector bundle $S$ of bounded geometry, define $EBD^m(X,S)$ as the space of differential operators $P$ such that in the exponential normal coordinates $y$ around $x\in X$, we have
\[
 P = \sum_{|\alpha|\leq m} a_{\alpha,x}(y)\partial_y^{\alpha},
\]
where $a_{\alpha,x}(y)$ are $N\times N$ matrices, identified with sections of $\mathrm{End}(S)$ in the local trivialisation of $S$ and a synchronous frame, and
\[
 \big\|\partial_y^{\beta} a_{\alpha,x}(0)\big\| \leq C_{\alpha,\beta},\qquad |\alpha|\leq m,
\]
where $\|\cdot\|$ is the norm on the fibre $\mathrm{End}(S)_x$ defined by the Hermitian metric $h.$

Given such an operator $P \in EBD^2(X,S)$, we say by analogy with the scalar-valued case that~$P$ has a density of states $\nu_P$ if for every $t$ there exists the limit
\[
 \lim_{R\to\infty} \frac{1}{|B(0,R)|}\mathrm{Tr}\big({\rm e}^{-tP}M_{\chi_{B(0,R)}}\big) = \int_{\mathbb{R}} {\rm e}^{-t\lambda}\,{\rm d}\nu_P(\lambda).
\]
Here, the trace is now with respect to the Hilbert space $L_2(X,S)$ rather than $L_2(X).$ A verbatim repetition of the proof of Theorem~\ref{T: main manifold thm} shows the following.
\begin{Theorem}\label{T: vector bundle dos thm}
 Let $S\to X$ be a vector bundle of bounded geometry over a non-compact Riemannian manifold of bounded geometry with Property~$($D$)$. If $P \in EBD^2(X,S)$ is a self-adjoint lower bounded operator having a density of states $\nu_P,$ then for $f\in C_c(\mathbb{R})$, we have
 \[
 \mathrm{Tr}_{\omega}(f(P)M_w) = \int_\mathbb{R} f(\lambda)\,{\rm d}\nu_P(\lambda),
 \]
 where $w(x) = (1+|B(x_0,d(x,x_0))|)^{-1}.$
 Similarly, for all $t>0$, we have
 \[
 \mathrm{Tr}_{\omega}(\exp(-tP)M_w) = \int_\mathbb{R} \exp(-t\lambda)\,{\rm d}\nu_P(\lambda).
 \]
\end{Theorem}

In~\cite{Roe1988a}, Roe considers (orientable) manifolds of bounded geometry that have a regular exhaustion. In this section, we will only consider manifolds that satisfy the assumptions of Theorem~\ref{T: main manifold thm}. In particular, the assumptions imply that $\lim_{R\to \infty}\frac{|\partial B(x_0, R)|}{|B(x_0, R)|} = 0$, i.e., any increasing sequence of metric balls $\{B(x_0, R_i)\}_{i=0}^\infty$ where $R\to\infty$ forms a regular exhaustion.

Denote the Banach space of $C^1$ uniformly bounded $n$-forms on $X$ by $\Omega^n_{\beta}(X)$. An element $m$ in the dual space of $\Omega^n_{\beta}(X)$ is said to be associated to the regular exhaustion $\{B(x_0, R_i)\}$ if for each bounded $n$-form~$\alpha$ \[\liminf_{i\to \infty} \bigg| \langle \alpha, m \rangle - \frac{1}{|B(x_0, R_i)|} \int_{B(x_0, R_i)} \alpha\bigg| = 0.\]
The algebra $\mathcal{U}_{-\infty}(X)$ consists of operators $A\colon C_c^\infty(X) \to C_c^\infty(X)$ such that for each $s,k\in \mathbb{R},$ $A$ has a continuous extension to a quasilocal operator from $H^s(X) \to H^{s-k}(X)$. In Roe's terminology, an operator $A\colon H^s(X) \to H^{s-k}(X)$ is quasilocal if for each $K \subset X$ and each $u\in H^k(X)$ supported within $K$,
\[
\| Au\|_{H^{s-k}(X\setminus \text{Pen}^+(K,r))} \leq \mu(r) \| u \|_{H^s(X)},
\]
where $\mu\colon \mathbb{R}^+ \to \mathbb{R}^+$ is a function such that $\mu(r) \to 0$ as $r \to \infty$, and $\text{Pen}^+(K,r)$ is the closure of $\cup \{B(x,r) \colon x\in K\}$.

Operators $A \in \mathcal{U}_{-\infty}(X)$ are represented by uniformly bounded smoothing kernels, i.e.,
\[
Au(x) = \int k_A(x,y)u(y) \text{vol}(y).
\]
Roe then defines traces on $\mathcal{U}_{-\infty}(X)$ coming from functionals associated to our regular exhaustion by{\samepage
\[
\tau(A) = \langle \alpha_A, m\rangle,
\]
where $\alpha_A$ is the bounded $n$-form defined by $x\to k_A(x,x)\text{vol}(x)$.}

Recall from the introduction that the trace $\tau$ extends to a trace on $M_n\big(\mathcal{U}_{-\infty}^+\big)$, and hence descends to a dimension-homomorphism
\[
\dim_\tau\colon \ K_0(\mathcal{U}_{-\infty}) \to \mathbb{R}.
\]
Furthermore, since Roe showed that elliptic differential operators are invertible modulo $\mathcal{U}_{-\infty}$~\cite{Roe1988a}, one can define an abstract index of an elliptic differential operator acting on sections of a Clifford bundle as an element of $K_0(\mathcal{U}_{-\infty})$ via standard $K$-theory constructions. The Roe index theorem then states the following.
\begin{Theorem}[Roe index theorem]
Let $X$ be a Riemannian manifold, $S$ a graded Clifford bundle on $X$, both with bounded geometry. Let $D$ be the Dirac operator of $S$. Let $m$ and $\tau$ be defined as above. Then
\[
 \dim_\tau (\operatorname{Ind}(D)) = m(\mathbf{I}(D)),
\]
where $\mathbf{I}(D)$ is the integrand in the Atiyah--Singer index theorem.
\end{Theorem}
The basis for the proof of the theorem is the following McKean--Singer formula
\[
 \dim_{\tau}(\operatorname{Ind}(D)) = \tau\big(\eta {\rm e}^{-tD^2}\big),
\]
where $\eta$ is the grading operator on $S,$ and $\tau$ is now defined on operators on acting on sections of $S$.

The following lemma relates Roe's $\tau$ functional to the density of states. We will use the fact that if $P$ is elliptic and self-adjoint, then the mapping
\[
 f\mapsto \tau(f(P))
\]
is continuous on $f \in C_c(\mathbb{R}).$ Indeed, by the Sobolev inequality the uniform norm of the integral kernel of $f(P)$ is bounded above by the norm of $f(P)$ as an operator from a Sobolev space of sufficiently negative smoothness to a Sobolev space of sufficiently positive smoothness. Since~$f$ is compactly supported, by functional calculus, the operator $f(P)(1+P)^N$ is bounded on~$L_2(X,S)$ for every $N,$ with norm depending on the width of the support of $f$ and the uniform norm of $f.$ Since $P$ is elliptic, it follows from these arguments that if $f$
is supported in $[-K,K]$ then there is a constant $C_K$ such that
\[
 |\tau(f(P))| \leq C_K\|f\|_{\infty}.
\]
See the related arguments in~\cite[Propositions 2.9 and~2.10]{Roe1988a}.

\begin{Lemma}\label{L: tau is dos}
 Let $S\to X$ be a vector bundle of bounded geometry,
 and let $P \in EBD^m(X,S)$ for some $m>0.$
 Assume that $P$ has a density of states $\nu_P$
 with respect to the base-point $x_0\in X.$ If $\tau$ is associated to the exhaustion $\{B(x_0,R_i)\}_{i=0}^\infty$ for some sequence $R_i\to\infty,$ then
 \[
 \tau(f(P)) = \int_{\mathbb{R}} f\,{\rm d}\nu_P,\qquad f \in C_c(\mathbb{R}).
 \]
 Similarly,
 \[
 \tau(\exp(-tP)) = \int_{\mathbb{R}} \exp(-t\lambda)\,{\rm d}\nu_P(\lambda),\qquad t > 0.
 \]
\end{Lemma}
\begin{proof}
By Theorem~\ref{T: vector bundle dos thm}, we have
\begin{align*}
\mathrm{Tr}_\omega (\exp(-tP) M_w)&= \int_{\mathbb{R}} {\rm e}^{-t\lambda} \,{\rm d}\nu_P(\lambda)\\
&= \lim_{R\to \infty} \frac{1}{|B(x_0, R)|} \mathrm{Tr}(\exp(-tP)M_{\chi_{B(x_0,R)}})\\
&= \lim_{R\to \infty} \frac{1}{|B(x_0, R)|} \int_{B(x_0,R)}\mathrm{tr}_{\mathrm{End}(S_x)}(K_{\exp(-tP)}(x,x)) \text{vol}(x)\\
&= \tau(\exp(-tP)).
\end{align*}
Since
\[
 f\mapsto \tau(f(P)),\qquad f \in C_c(\mathbb{R})
\]
is continuous in the sense described in the paragraph preceding the theorem, it follows from the Riesz theorem that there exists a measure $\mu_{\tau,P}$ on $\mathbb{R}$ such that
\[
 \tau(f(P)) = \int_{\mathbb{R}} f\,{\rm d}\mu_{\tau,P},\qquad f \in C_c(\mathbb{R}).
\]
Since $\mu_{\tau,P}$ and $\nu_P$ have identical Laplace transform, it follows that $\mu_{\tau,P} = \nu_P.$
\end{proof}

A combination of Lemma~\ref{L: tau is dos} and Theorem~\ref{T: vector bundle dos thm} immediately yields the following:
\begin{Theorem}
Let $X$ be a manifold that satisfies the assumptions of Theorem~$\ref{T: main manifold thm}$, and let $S\to X$ be a vector bundle of bounded geometry. Let $P \in EBD^2(X,S)$, be self-adjoint and lower-bounded, and assume it admits a density of states $\nu_P$ at $x_0$. Let $w$ be the function on $X$ defined by
\[
w(x) = (1+\abs{B(x_0,d_X(x,x_0))})^{-1},\qquad x \in X.
\]
Then for any $f\in C_c(\mathbb{R})$, we have
\[
\tau(f(P)) = \mathrm{Tr}_\omega(f(P)M_w)
\]
for any $\tau$ associated to the regular exhaustion $\{B(x_0,R_i)\}_{i\in \mathbb{N}}$ where $R_i\to\infty,$ and for any extended limit $\omega.$
Similarly,
\[
 \tau(\exp(-tP)) = \mathrm{Tr}_\omega(\exp(-tP)M_w).
\]
\end{Theorem}

The preceding theorem is proved under the strong assumption that $P$ admits a density of states, which in particular implies that $\tau(\exp(-tP))$ is independent of the choice of functional~$m$ used to define $\tau.$ Addressing the question of determining which traces $\tau$ and which extended limits $\omega$ are related in this way in general is beyond the scope of this article.

Roe~\cite{Roe1988a} defines an algebra $\mathcal{U}_{-\infty}(E)$ of operators acting on sections of a vector bundle $E,$ and $\tau$ is extended to $\mathcal{U}_{-\infty}(E)$ essentially by composing $\tau$ with the pointwise trace on $\mathrm{End}(E),$ see~\cite{Roe1988a} for details.

\begin{Theorem}\label{roe_index_formula_theorem}
Let $X$ be a non-compact Riemannian manifold of bounded geometry with Property~$($D$)$, with a graded Clifford bundle $S \to X$ of bounded geometry. Let $D$ be a Dirac operator on $X$ associated with the Dirac complex
\[
C^\infty\big(S^+\big) \xrightarrow{D_+} C^\infty(S^-),
\]
where $D_+$ is the restriction of $D$ to the sections of $S^+$ and $D_- = D_+^*$ is its adjoint $($cf.~{\rm \cite[Chapter~11]{Roe1998})}.
 Let $D^2$ admit a density of states both when considered as an operator restricted to $L_2\big(S^+\big)$ and $L_2(S^-)$, in the sense that
 \[
 \lim_{R\to \infty} \frac{1}{|B(x_0,R)|} \mathrm{Tr}(\exp(-tD_-D_+)M_{\chi_{B(x_0,R)}}) = \int_{\mathbb{R}} {\rm e}^{-t\lambda} {\rm d}\nu_{D_-D_+}(\lambda), \qquad t>0,
 \]
 for a Borel measure $\nu_{D_-D_+}$ and similarly for $D_+D_-$. Then for any $f\in C_c(\mathbb{R})$ such that $f(0)=1$, we have
\[
\dim_\tau(\mathrm{Ind} D) = \mathrm{Tr}_\omega\big(\eta f\big(D^2\big)M_w\big).
\]
\end{Theorem}
\begin{proof}
By~\cite[Proposition~8.1]{Roe1988a}, we have that
\[
\dim_\tau(\text{Ind} D) = \tau\bigl(\eta \exp\bigl(-tD^2\bigr)\bigr), \qquad t>0,
\]
where $\eta$ is the grading operator
\[
 \eta = \begin{pmatrix} 1 & \hphantom{-} 0 \\ 0 & -1\end{pmatrix}
\]
with respect to the orthogonal direct sum $L_2(S) = L_2\big(S^+\big)\oplus L_2(S^-).$

The proof of the theorem amounts to showing that
\[
 \tau\big(\eta {\rm e}^{-tD^2}\big) = \mathrm{Tr}_{\omega}\big(\eta {\rm e}^{-tD^2}M_w\big).
\]
The left-hand side is the same as
\[
 \tau\big({\rm e}^{-tD_+D_-}\big)-\tau\big({\rm e}^{-tD_-D_+}\big)
\]
while the right-hand side is
\[
 \mathrm{Tr}_{\omega}\big({\rm e}^{-tD_+D_-}M_w\big)-\mathrm{Tr}_{\omega}\big({\rm e}^{-tD_-D_+}M_w\big).
\]
Applying Theorem~\ref{T: vector bundle dos thm} to $D_+D_-$ and $D_-D_+$ individually proves the result.
\end{proof}

\begin{Remark}
The index $\dim_{\tau}(\operatorname{Ind}(D))$ is computed by a version of the McKean--Singer formula~\cite[Proposition 8.1]{Roe1988a}
\[
 \mathrm{dim}_{\tau}(\mathrm{Ind}(\mathrm{D})) = \tau\big(\eta \exp\bigl(-tD^2\bigr)\big)
\]
for arbitrary $t>0.$ One of the motivations in developing the present theorem was
to give a new explanation of why the function
\[
 t\mapsto\tau\big(\eta\exp\bigl(-tD^2\bigr)\big)
\]
is independent of $t.$ If the assumptions of Theorem~\ref{roe_index_formula_theorem} hold, then
\[
 \tau\big(\eta \exp\bigl(-tD^2\bigr)\big) = \mathrm{Tr}_\omega\big(\eta {\rm e}^{-tD^2}M_w\big).
\]
Formally differentiating the right-hand side with respect to $t$ and using the tracial property of~$\mathrm{Tr}_\omega$ yields
\[
 \frac{{\rm d}}{{\rm d}t}\mathrm{Tr}_\omega\big(\eta {\rm e}^{-tD^2}M_w\big) = -\mathrm{Tr}_{\omega}\big(\eta {\rm e}^{-tD^2}D[D,M_w]\big).
\]
Our conditions ensure that ${\rm e}^{-tD^2}D[D,M_w]$ is trace class, and hence that
\[
 \frac{{\rm d}}{{\rm d}t}\mathrm{Tr}_{\omega}\big(\eta {\rm e}^{-tD^2}M_w\big) = 0.
\]
It is interesting that $\mathrm{Tr}_{\omega}\big(\eta {\rm e}^{-tD^2}M_w\big)$ and the traditional heat supertrace $\mathrm{Tr}\big(\eta {\rm e}^{-tD^2}\big)$ on a compact manifold are both independent of $t$ for apparently different reasons.
\end{Remark}

\section{An example with a random operator}\label{S: Example}
The assumptions in Theorem~\ref{T: main manifold thm} appear quite strong, especially the existence of the density of states. The following example of a random operator on a non-compact manifold where the density of states exists was given by Lenz, Peyerimhoff and Veseli\'{c}~\cite{LenzPeyerimhoffVeselic2007}, generalising earlier examples in~\cite{LenzPeyerimhoff2004,PeyerimhoffVeselic2002}.

\begin{Example}[{\cite[Example (RSM)]{LenzPeyerimhoffVeselic2007}}]
 Let $(X,g_0)$ be the connected Riemannian covering of a compact Riemannian manifold $M = X/\Gamma,$ where $\Gamma$ is an infinite group acting freely and properly discontinuously on $X$ by isometries. Assume that there is an ergodic action $\alpha$ of $\Gamma$ on a probability space $(\Omega,\Sigma,\mathbb{P}),$ and let $\{g_\omega\}_{\omega\in \Omega}$ be a measurable family of metrics on $X$ which are uniformly comparable with $g_0,$ in the sense that there exists $A>0$ such that
 \[
 \frac1A g_0(v,v) \leq g_\omega(v,v)\leq Ag_0(v,v).
 \]
 for all tangent vectors $v$ to $X.$ Assume that the action $\alpha$ of $\Gamma$ on $\Omega$ is compatible with the action on $\Gamma$ on $X$ in the sense that $g_{\alpha_{\gamma}^{-1}\omega}$ is the pullback of $g_{\omega}$ under the automorphism defined by $\gamma.$

 Similarly, it is assumed that there is a measurable family $\{V_{\omega}\}_{\omega\in \Omega}$ of smooth functions on $X$ such that
 \[
 V_{\omega}\circ \gamma = V_{\alpha^{-1}_\gamma\omega}.
 \]
 Let $\nu^\omega$ denote the Riemannian volume form on $X$ corresponding to $g_{\omega},$ and let $\Delta_{\omega}$ be the Laplace--Beltrami operator on $(X,g_{\omega}).$ Then~\cite{LenzPeyerimhoffVeselic2007} consider the operator on $L_2(X,\nu^{\omega})$ given by
 \[
 H_{\omega} = -\Delta_{\omega}+M_{V_\omega}.
 \]
 What is shown in~\cite[equation~(27)]{LenzPeyerimhoffVeselic2007} is that if $\Gamma$ is amenable, then for every tempered F{\o}lner sequence $\{A_n\}_{n=0}^\infty$ in $\Gamma,$ the limit
 \[
 \lim_{n\to\infty} \frac{1}{|A_n|}\mathrm{Tr}_{L_2(X,\nu^{\omega})}\big(M_{\chi_{A_nF}}{\rm e}^{-tH_{\omega}}\big)
 \]
 exists for almost every $\omega,$ where $F$ is a fundamental domain for $\Gamma.$

 Recall that we say that a F{\o}lner sequence $\{A_n\}_{n=0}^\infty$ is tempered if there is a constant $C>0$ such that for every $n\geq 1$, we have
 \[
 \bigg|\bigcup_{k<n} A_k^{-1}A_{n}\bigg| \leq C|A_{n}|.
 \]
\end{Example}

We will make one further assumption: that the measure $\nu^{0}$ associated to $g_0$ has a doubling property. That is, there exists a constant $C$ such that
\[
 \nu^0(B_{g_{0}}(x_0,2R)) \leq C\nu^{0}(B_{g_{0}}(x_0,R)),\qquad R>0.
\]
Note that since the identity function is a bi-Lipschitz continuous map from $(X,g_0)$ to $(X,g_{\omega}),$ the same holds for the measures $\nu^{\omega}$ associated to the metrics $g_{\omega}.$

\begin{Proposition}
 Let $(X, g_0)$ be as above. If $X$ satisfies Property~$($D$)$, the density for $H_\omega$ exists almost surely, and therefore
 \[
 \mathrm{Tr}_{\varpi}(f(H_\omega)M_w) = \int_{\mathbb{R}} f(\lambda)\, {\rm d}\nu_{H_\omega}(\lambda), \qquad f\in C_c(\mathbb{R})
 \]
 for every extended limit $\varpi.$
\end{Proposition}
\begin{proof}
 For brevity, we will denote the measure $\nu^{\omega}$ by $|\cdot|$
 and $B(x_0,R)$ for $B_{g_{\omega}}(x_0,R).$

 We will show that Property~(D) implies that $\Gamma$ admits a tempered F{\o}lner sequence $\{A_n\}_{n=1}^\infty$, and that for any bounded measurable function $g$ on $X$, we have
 \begin{equation}\label{equivalence_of_limits}
 \lim_{k\to\infty} \frac{1}{|A_kF|} \int_{A_kF}g\,{\rm d}\nu^{\omega} = \lim_{R\to\infty} \frac{1}{|B(x_0,R)|}\int_{B(x_0,R)} g\,{\rm d}\nu^{\omega}
 \end{equation}
 if either limit exists. Recall that $F$ is a fundamental domain for the action of $\Gamma.$
 Together with the results of~\cite{LenzPeyerimhoffVeselic2007}, this implies that the limit
 \[
 \lim_{R\to\infty} \frac{1}{|B(x_0,R)|}\mathrm{Tr}\big(M_{\chi_{B(x_0,R)}}{\rm e}^{-tH_{\omega}}\big)
 \]
 exists for every $t>0$, and hence the assumptions of Theorem~\ref{T: main manifold thm} are satisfied.

 Let $h>2\mathrm{diam}(F).$ For $k\geq h,$ let
 \[
 A_k = \{\gamma \in \Gamma\colon \mathrm{dist}(\gamma x_0,x_0) < k-h \}.
 \]
 We claim that $A_k$ is a tempered F{\o}lner sequence.

 Note that $\mathrm{dist}(\gamma x_0,x_0) = \mathrm{dist}\big(\gamma^{-1}x_0,x_0\big),$ and so automatically $A_k=A_k^{-1}.$
 Define
 \[
 B_k := A_kF = \bigcup_{\gamma\in A_k} \gamma F.
 \]
 First, we show that $B(x_0,k-2h)\subseteq B_k.$ Indeed, if $p \in B(x_0,k-2h)$ there exists some $\gamma\in \Gamma$ such that $p \in \gamma F,$ so $\mathrm{dist}(p,x_0) \leq \mathrm{diam}(F) < h,$ and thus $\mathrm{dist}(\gamma x_0,x_0) < k-2k+h =k-h.$ On the other hand, since $F$ has diameter smaller than $\frac{h}{2},$ if $p\in B_k$ then $\mathrm{dist}(p,x_0)\leq \frac{h}{2}+k-h < k.$
 That is, for all $k\geq 2h$, we have
 \[
 B(x_0,k-2h)\subset B_k \subset B(x_0,k).
 \]
 Since the action of $\Gamma$ is free, the union of the translates of $F$ is disjoint, and
 \[
 |B_k| = |A_k||F|,\qquad k\geq 0
 \]
 and hence
 \[
 |B(x_0,k-2h)| \leq |F||A_k| \leq |B(x_0,k)|.
 \]
 As in the proof of Lemma~\ref{Grimaldi}, we know that for each $h>0$ there is a constant $C_h$, so that we have
 \[
 \frac{|B(x_0, k)|}{|B(x_0, k-2h)|} \leq C_h, \qquad k > 1+2h.
 \]
 Therefore,
 \[
 |A_k| \approx |B_k| \approx |B(x_0,k)|
 \]
 uniformly in $k>1+2h.$

 To see that $A_k$ is F{\o}lner, let $\gamma\in \Gamma.$ By the triangle inequality, there exists $N>0$ such that
 \[
 \gamma A_k\subset A_{k+N}
 \]
 and, for $k$ sufficiently large,
 \[
 A_k\subset \gamma A_{k-N}
 \]
 and therefore the symmetric difference of $A_k$ and $\gamma A_k$ satisfies
 \[
 (\gamma A_k\setminus A_k)\cup (A_k\setminus \gamma A_k) \subset A_{k+N}\setminus A_{k-N}.
 \]
 It follows that, as $k\to\infty,$
 \[
 \frac{|(\gamma A_k\setminus A_k)\cup (A_k\setminus \gamma A_k)|}{|A_k|} \approx \frac{|B(x_0,k+N)|-|B(0,k-N)|}{|B(x_0,k)|},
 \]
 which is vanishing as $k\to\infty.$ Hence, $\{A_k\}_{k=1}^\infty$ is F{\o}lner.

 To see that $\{A_k\}_{k=1}^\infty$ is tempered, it suffices to show that there is a constant $C$ such that
 \[
 |A_k\cdot A_k| \leq C|A_k|.
 \]
 By the triangle inequality and the fact that $\Gamma$ acts isometrically, we see that
 \[
 A_k\cdot A_k \subseteq A_{2k}.
 \]
 Therefore, by the doubling condition,
 \[
 |A_k\cdot A_k| \leq |A_{2k}| = \frac{1}{|F|}|B(x_0,2k)| \lesssim |B(x_0,k)| \approx |A_k|
 \]
 uniformly in $k$ for sufficiently large $k.$
 Hence, $\{A_k\}_{k=1}^\infty$ is tempered.

Finally, we prove~\eqref{equivalence_of_limits}. Using the fact that $B_k\subseteq B(x_0,k)$, we write
\begin{align*}
 \frac{1}{|B(x_0,k)|}\int_{B(x_0,k)} g\,{\rm d}\nu^{\omega}={}& \frac{1}{|B_k|}\int_{B_k} g\,{\rm d}\nu^{\omega} + \frac{|B_k|-|B(x_0,k)|}{|B(x_0,k)||B_k|}\int_{B_k} g\,{\rm d}\nu^{\omega} \\
 &{} + \frac{1}{|B(x_0,k)|}\int_{B(x_0,k)\setminus B_k} g\,{\rm d}\nu^{\omega}.
\end{align*}
Therefore,
\begin{align*}
 \bigg|\frac{1}{|B(x_0,k)|}\int_{B(x_0,k)} g\,{\rm d}\nu^{\omega} - \frac{1}{|B_k|}\int_{B_k} g\,{\rm d}\nu^{\omega}\bigg|
 &\leq 2\|g\|_{\infty}\frac{|B(x_0,k)|-|B_k|}{|B(x_0,k)|}\\
 &\leq 2\|g\|_{\infty} \frac{|B(x_0,k)|-|B(x_0,k-2h)|}{|B(x_0,k)|}
\end{align*}
and this vanishes as $k\to\infty.$ From here one easily deduces~\eqref{equivalence_of_limits}.
\end{proof}

\subsection*{Acknowledgements}

We are grateful to Fedor Sukochev for his support and encouragement, and to Teun van Nuland for a helpful discussion. We thank the anonymous referees for comments that helped improve the paper.

\pdfbookmark[1]{References}{ref}
\LastPageEnding

\end{document}